\title{Multiple Access Channel Resolvability Codes from Source Resolvability Codes}
\author{Rumia Sultana and R\'{e}mi A. Chou\\
\thanks{Part of this work has been presented at the 2019 Annual Allerton Conference on Communication, Control, and Computing~\cite{sultana2019explicit}, and the 2020 IEEE International Symposium on Information Theory \cite{sultana2020explicit}. This work was supported in part by NSF grant CCF-1850227. E-mails: rxsultana@shockers.wichita.edu; remi.chou@wichita.edu. } }
\documentclass[lettersize,journal]{IEEEtran}
\IEEEoverridecommandlockouts

\usepackage{epsfig,latexsym,amsmath,amssymb,setspace,cite,color,graphicx,algorithm,algorithmic,cases}
\usepackage[font=footnotesize]{subfig}
\usepackage[percent]{overpic}
\usepackage{amsmath}
\usepackage{amsfonts}
\usepackage{amssymb}
\usepackage{dsfont}
\usepackage{bm}
\usepackage{amsthm}
\usepackage{newlfont}
\usepackage{float}
\usepackage{hyperref}
\usepackage{algorithm}
\usepackage{algorithmic}
\usepackage{enumerate}
\usepackage{chngcntr}
\usepackage{mathtools}
\usepackage{breqn}
\usepackage{stmaryrd}
\usepackage{cite}
\usepackage{amsmath}
\usepackage{amsfonts}
\usepackage{amssymb}
\usepackage{dsfont}
\usepackage{bbm}
\usepackage{amsthm}
\usepackage{newlfont}
\usepackage{float}
\usepackage{hyperref}
\usepackage{enumerate}
\usepackage{chngcntr}
\usepackage{mathtools}
\usepackage{breqn}
\usepackage{stmaryrd}
\usepackage{cite}
\usepackage{algorithm}
\usepackage{algorithmic}
\usepackage{caption}
\hypersetup{
    colorlinks=true, 
    linkcolor= blue,
    citecolor=blue 
}
\newtheorem{definition}{Definition}
\newtheorem{lem}{Lemma}
\newtheorem{thm}{Theorem}
\newtheorem{corollary}{Corollary}

\allowdisplaybreaks
\newcommand{\argmax}{\operatornamewithlimits{argmax}}
\begin{document}
\maketitle
\begin{abstract}
We show that the problem of code construction for multiple access channel (MAC) resolvability can be reduced to the simpler problem of code construction for source resolvability. Specifically, we propose a MAC resolvability code construction that relies on a combination of multiple source resolvability codes, used in a black-box manner, and leverages randomness recycling implemented via distributed hashing and block-Markov coding. 
Since explicit source resolvability codes are known, our results also yield the first explicit coding schemes that achieve the entire MAC resolvability region for any discrete memoryless multiple-access channel with binary input alphabets.  
\end{abstract}
\section{Introduction}
The concept of multiple access channel (MAC) resolvability has been introduced in \cite{steinberg1998resolvability} as a natural extension of channel resolvability for point-to-point channels \cite{han1993approximation}. MAC resolvability represents a fundamental primitive that finds applications in a large variety of network information-theoretic problems, including strong secrecy for  multiple access wiretap channels~\cite{pierrot2011strongly,yassaee2010multiple}, cooperative jamming~\cite{pierrot2011strongly}, semantic security for multiple access wiretap channels \cite{frey2017mac}, and strong coordination in networks~\cite{bloch2013strongco}. These applications are, however, restricted by
the fact that   no explicit coding scheme is known to optimally implement MAC resolvability. Note indeed that \cite{steinberg1998resolvability,frey2017mac} only provide existence results and no explicit code constructions. The objective of this paper is to bridge this gap by providing explicit coding schemes that achieve the MAC resolvability region \cite{frey2017mac}. 
While previous works have been successful in providing explicit coding schemes for channel resolvability over point-to-point channels,\footnote{Explicit constructions based on polar codes for channel resolvability have been proposed for binary \emph{symmetric} point-to-point channels~\cite{bloch2012strong} and 
discrete memoryless point-to-point channels whose input alphabets have prime cardinalities~\cite{chou2018empirical}. Another explicit construction based on injective group homomorphisms has been proposed in~\cite{hayashi2016secure} for channel resolvability over binary \emph{symmetric} point-to-point channels. Low-complexity, but non-explicit, linear coding schemes for channel resolvability  over arbitrary memoryless point-to-point channels have also been proposed in \cite{amjad2015channel}.} to the best of our knowledge, the only known explicit constructions for MAC resolvability are those of \cite{chou2014low}. However, the explicit constructions in~\cite{chou2014low}, one based on invertible extractors and a second one based on injective group homomorphisms, are limited to \emph{symmetric} multiple access channels, and do not seem to
generalize to \emph{arbitrary} multiple access channels. 

In this paper, we propose a novel approach to the construction of MAC resolvability codes by showing that such a construction can be reduced to the simpler problem of code construction for source resolvability~\cite{han2003information}. Since explicit constructions of source resolvability codes are known, e.g.,~\cite{chou2018empirical}, our results yield the first explicit construction of MAC resolvability codes that achieve the entire MAC resolvability region of arbitrary multiple access channels with binary input alphabets.
More specifically, our approach to the construction of MAC resolvability codes   relies on a combination of appropriately chosen source resolvability codes,  and leverages randomness recycling implemented with distributed hashing and a block-Markov encoding scheme. 
In essence, the idea of block-Markov encoding to recycle randomness is closely related to recursive constructions of seeded extractors in the computer science literature, e.g., \cite{vadhan2012pseudorandomness}. We stress that our construction is valid independently from the way those source resolvability codes are implemented.  Additionally, \textcolor{black}{to avoid time-sharing whenever it is known
to be unnecessary, we also show how to implement the idea of rate splitting, first developed
in  \cite{grant2001rate} for multiple access channel coding, for the MAC resolvability
problem with two transmitters.} 
Note that the main difference with~\cite{chou2014low}, is that our approach aims to reduce the construction of MAC resolvability codes to a simpler problem, namely the construction of source resolvability codes, whereas~\cite{chou2014low} attempts a code construction directly adapted to multiple access channels. 

The remainder of the paper is organized as follows. The problem statement is provided in Section~\ref{sec:problem statement}. Our main result is summarized in Section \ref{sec:main}. Our proposed coding scheme and its analysis are provided in Section~\ref{sec:coding scheme} and Section~\ref{sec:coding scheme analysis}, respectively. While our main result focuses on  multiple access channels with two transmitters, we discuss an extension of our result to an arbitrary number of transmitters in Section~\ref{sec:extension}. Finally, Section~\ref{sec:conclusion} provides concluding remarks.
\section{Problem Statement  and Review of Source Resolvability}\label{sec:problem statement}
\subsection{Notation} \label{sec:notations} For $a,b\in \mathbb{R}$, define $\llbracket a,b\rrbracket \triangleq [ \lfloor a \rfloor, \lceil b \rceil]\cap \mathbb{N}$.
 The components of a vector $X^{1:N}$ of size $N$ are denoted with superscripts, i.e., $X^{1:N} \triangleq (X^1,X^2,\ldots,X^{N})$.  For two  probability  distributions $p$ and $q$ defined  over  the  same alphabet
$\mathcal X$, the variational distance $\mathbb{V}(p,q)$ between $p$ and $q$ is defined
as $\mathbb{V}(p,q)\triangleq \sum_{x \in \mathcal{X}}| p(x)-q(x)|.$
\subsection{Problem Statement} \label{sec:ps}
Consider a discrete memoryless multiple access  channel  $(\mathcal X\times \mathcal Y,  q_{Z|XY}, \mathcal Z)$, where $\mathcal X=\{ 0,1\}=\mathcal Y$, and $\mathcal{Z}$ is a finite alphabet. A target distribution $q_Z$ is defined as the  channel output distribution when the input distributions are $q_X$ and $q_Y$, i.e., 
\begin{align}
 \forall z \in \mathcal{Z}, q_Z (z) \triangleq \displaystyle \sum_{x \in \mathcal{X}} \sum_{y\in \mathcal{Y}} q_{Z|XY}(z|x,y) q_X(x)q_Y(y).\label{eqn1}
\end{align}
\begin{definition}\label{definition1}
 A $(2^{NR_1}, 2^{NR_2}, N)$  code for the memoryless multiple access channel $(\mathcal X\times \mathcal Y,  q_{Z|XY}, \mathcal Z)$ consists of 
\begin{itemize}
\item Two randomization sequences $S_1$ and $S_2$ independent and uniformly distributed over $\mathcal{S}_1 \triangleq \llbracket 1, 2^{NR_1} \rrbracket$ and $\mathcal{S}_2\triangleq~ \llbracket 1, 2^{NR_2} \rrbracket$, respectively;
\item Two encoding functions $f_{1,N}: \mathcal{S}_1 \to \mathcal{X}^N$ and $f_{2,N}: \mathcal{S}_2 \to \mathcal{Y}^N$;
\end{itemize}
and operates as follows: Transmitters 1 and 2 form $f_{1,N}(S_1)$ and $f_{2,N}(S_2)$, respectively, which are sent over the channel $(\mathcal X\times \mathcal Y , q_{Z|XY}, \mathcal Z)$.
\end{definition}
\begin{definition}\label{definition2}
$(R_1, R_2)$ is an achievable resolvability rate pair  for the memoryless  multiple access channel $(\mathcal X\times \mathcal Y,  q_{Z|XY}, \mathcal Z)$ if there exists  a sequence of $(2^{NR_1}, 2^{NR_2}, N)$  codes such that 
$$\displaystyle\lim_{N \to +\infty } \mathbb{V} (\widetilde{p}_{Z^{1:N}},q_{Z^{1:N}}) = 0,
$$
where $q_{Z^{1:N}} \triangleq \prod_{i=1}^N q_Z$ with $q_Z$ defined in $(1)$ and $\forall z^{1:N} \in \mathcal{Z}^N, $
\begin{align*}
 \widetilde{p}_{Z^{1:N}} (z^{1:N})  \triangleq \!\!\!\!\!\!\!\!\!\!\!\! \sum_{(s_1,s_2) \in \mathcal{S}_1 \times \mathcal{S}_2} \!\!\!\!\!\!\!\!\!\!\!\! \frac{ q_{Z^{1:N}|X^{1:N}Y^{1:N}}\!\left(z^{1:N}| f_{1, N}(s_1), 
f_{2, N}(s_2)  \right) } {|\mathcal{S}_1||\mathcal{S}_2|}.
\end{align*}
The multiple access channel resolvability region $\mathcal R_{q_Z}$ is defined as the closure of the set of all achievable rate pairs.
\end{definition}
 \begin{thm}[{\cite[Theorem 1]{frey2017mac}}]
 We have $\mathcal R_{q_Z} = \mathcal{R}'_{q_Z}$ with
 \begin{align*}
   \mathcal R'_{q_Z} \triangleq \smash{\bigcup_{p_T,q_{X|T},q_{Y|T}} }\!\!\!\!\!\!\!\!\!\{(R_1, R_2):  I(X Y;Z|T)&\leq R_1+R_2,\\
  I(X;Z|T)&\leq R_1, \\
   I(Y;Z|T)&\leq R_2\},
 \end{align*}
    where $p_T$ is defined  over  $\mathcal{T}\triangleq \llbracket1,|\mathcal{Z}|+3\rrbracket$ and  $q_{X|T}, q_{Y|T}$ are such that, for any $t\in \mathcal{T}$ and  $z\in \mathcal{Z}$,
    \begin{align*}
        q_Z (z)= \sum_{{x \in \mathcal{X}}}\sum_{y \in \mathcal{Y}} q_{X|T}(x|t)q_{Y|T}(y|t)q_{Z|XY}(z|x,y).
    \end{align*}\label{thm1}
        \end{thm}
        \vspace*{-1em}
Note that reference \cite{frey2017mac} provides only  the existence of a coding scheme that achieves any rate pair in $\mathcal R_{q_Z}$. By contrast, \emph{our goal is to provide explicit coding schemes that can achieve the region $\mathcal R_{q_Z}$} by relying on source resolvability codes, which are used in a black box manner. The notion of source resolvability is reviewed next.
\subsection{Review of Source Resolvability}
\begin{definition}  \label{definition3}
A $(2^{NR},N)$ source resolvability code for $(\mathcal{X},q_{X})$ consists of
\begin{itemize}
\item A randomization sequence $S$ uniformly distributed over $\mathcal{S}  \triangleq \llbracket 1, 2^{NR} \rrbracket$;
\item An encoding function $e_N: \mathcal{S} \to \mathcal{X}^N$;
\end{itemize}
and operates as follows: The encoder forms $ \widetilde{X}^{1:N} \triangleq e_N(S)$ and the  distribution of $\widetilde{X}^{1:N}$ is denoted by $\widetilde{p}_{X^{1:N}}$.
\end{definition}
\begin{definition} \label{definition4}
$R$ is an achievable resolution rate for a discrete memoryless source $(\mathcal{X}, q_{X})$ if there exists a sequence of $(2^{NR},N)$ source resolvability codes such that
\begin{align}
    \lim_{N \to +\infty } \mathbb{V} (\widetilde{p}_{X^{1:N}},q_{X^{1:N}}) = 0,\label{eqndef4}
\end{align}
where $q_{X^{1:N}}\triangleq \prod_{i=1}^N q_X$. The infimum of such achievable rates is called  source resolvability.
 \end{definition}
\begin{thm}[\!\!\cite{han1993approximation}]
The source resolvability of a discrete memoryless source $(\mathcal{X}, q_X)$ is $H(X)$.\label{thm2}
\end{thm}

Note that explicit low-complexity source resolvability codes can, for instance, be obtained with polar codes as reviewed in Appendix \ref{Appsc}.

\section{Main result} \label{sec:main}
Our main result is summarized as follows.
\begin{thm} \label{tmain}
The coding scheme presented in Section \ref{sec:coding scheme}, which solely relies on source resolvability codes, used as black boxes, and two-universal hash functions \cite{carter1979universal}, achieves the entire  multiple access channel resolvability region $\mathcal R_{q_Z}$ for any discrete memoryless multiple access channel with binary input alphabets. 
\textcolor{black}{Moreover, time-sharing
is avoided whenever it is known to be unnecessary.}
\end{thm}

As a corollary, we obtain the first explicit construction of multiple access channel resolvability codes that achieves the entire  multiple access channel resolvability region $\mathcal R_{q_Z}$ for any discrete memoryless multiple access channel with binary input alphabets. 

\begin{corollary}
Since explicit constructions for source resolvability codes and two-universal hash functions are known, e.g., \cite{carter1979universal,chou2015polar}, Theorem \ref{tmain} yields an explicit coding scheme that achieves $\mathcal R_{q_Z}$ for any discrete memoryless multiple access channel with binary input alphabets.  
\end{corollary}
\section{Coding Scheme}\label{sec:coding scheme}
We explain in Section~\ref{sec:cases} that the general construction of  MAC resolvability codes can be reduced to two special cases. Then, we provide a coding scheme for these two special cases in Sections~\ref{sec:sub1},~\ref{sec:sub2}.
\subsection{Reduction of the general construction of MAC resolvability codes to two special cases}  \label{sec:cases}
\begin{definition}\label{definition5}
For the memoryless  multiple access channel $(\mathcal X\times \mathcal Y,  q_{Z|XY}, \mathcal Z)$ we define
 \begin{align*}
\mathcal R_{X,Y}\triangleq  \{(R_1,R_2):  I(XY;Z)& \leq R_1+R_2,\\
   I(X;Z)& \leq R_1, \\
   I(Y;Z)& \leq R_2
    \},
\end{align*}
for some product distribution $p_Xp_Y$ on $\mathcal X\times \mathcal Y$.
 \end{definition} 
 To show the achievability of $\mathcal R'_{q_{Z}}$, it is sufficient to show the achievability of $\mathcal R_{X,Y}$. Indeed,  note that if $\mathcal R_{X,Y}$ is achievable, then  $\textup{Conv}(\bigcup_{p_Xp_Y}\mathcal R_{X,Y})$ is also achievable, where $\textup{Conv}$ denotes the convex hull. Hence, $\mathcal R'_{q_{Z}}$ is achievable because   $\textup{Conv}(\bigcup_{p_Xp_Y}\mathcal R_{X,Y}) \supset \mathcal{R}'_{q_{Z}}$ by remarking that the corner points of $\mathcal{R}'_{q_{Z}}$  are in $\textup{Conv}(\bigcup_{p_Xp_Y}\mathcal R_{X,Y})$. For instance, the point $(I(X;Z|T), I(Y;Z|XT)) \in \mathcal{R}'_{q_{Z}}$ belongs to $\textup{Conv}(\bigcup_{p_Xp_Y}\mathcal R_{X,Y})$ since 
 \begin{align*}
    &(I(X;Z|T), I(Y;Z|XT)) \\&=  \sum_{t\in \mathcal{T}}p_T(t)(I(X;Z|T=t), I(Y;Z|X,T =t)).
 \end{align*} Similarly, all the corner points of $\mathcal{R}'_{q_{Z}}$ also belong to $\textup{Conv}(\bigcup_{p_Xp_Y}\mathcal R_{X,Y})$.
 Next, we consider two cases to achieve the region~$\mathcal R_{X,Y}$ for some fixed distribution $p_{X}p_Y$.
\begin{itemize}
\item Case 1 (depicted in Figure \ref{fig:model}): $I(XY;Z)>I(X;Z)+ I(Y;Z)$. In this case, it is sufficient to achieve the dominant face $\mathcal{D}$ of $\mathcal{R}_{X,Y}$, where
  \begin{align*}\mathcal D\triangleq\{(R_1,R_2): R_1  \in & [ I(X;Z), I(X;Z|Y)], \\ R_2  = & I(XY;Z)-R_1\}.\end{align*}
\item Case 2 (depicted in Figure \ref{fig:model2}): $I(XY;Z) = I(X;Z)+ I(Y;Z)$. In this case, only the corner point $C$ needs to be achieved. Note that it is impossible to have $I(XY;Z) < I(X;Z)+ I(Y;Z)$ by independence of $X$ and $Y$.
\end{itemize}
\begin{figure}
\begin{minipage}{.5\textwidth}
       \captionsetup{width=6 cm}
  \includegraphics[width=7.5 cm]{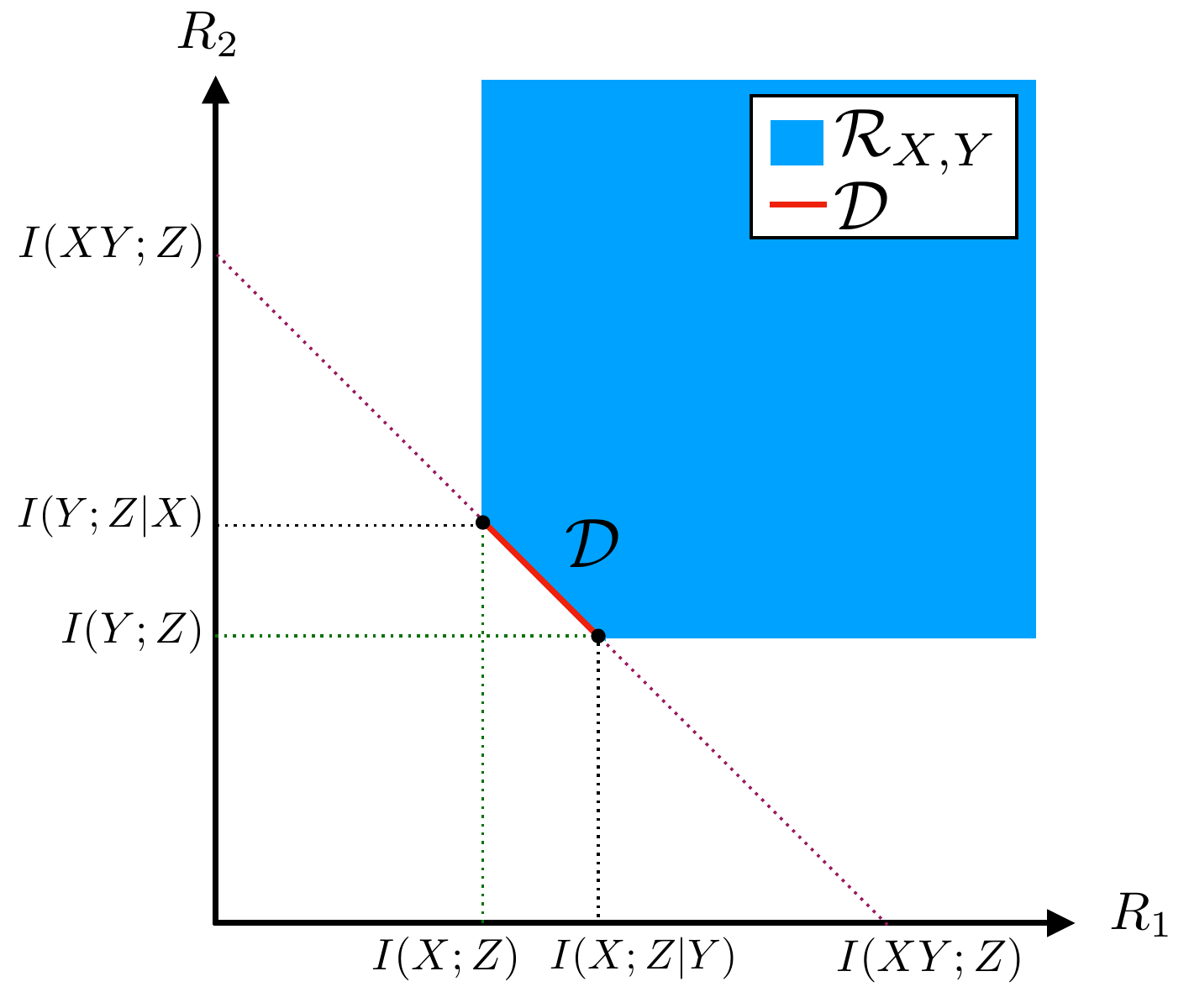}
  \caption{Region $\mathcal{R}_{X,Y}$ in Case 1: $I(XY;Z)>~I(X;Z)+ I(Y;Z)$.}
  \label{fig:model}
\end{minipage}
\begin{minipage}{.5\textwidth}
      \captionsetup{width=6 cm}
  \includegraphics[width=7.5 cm]{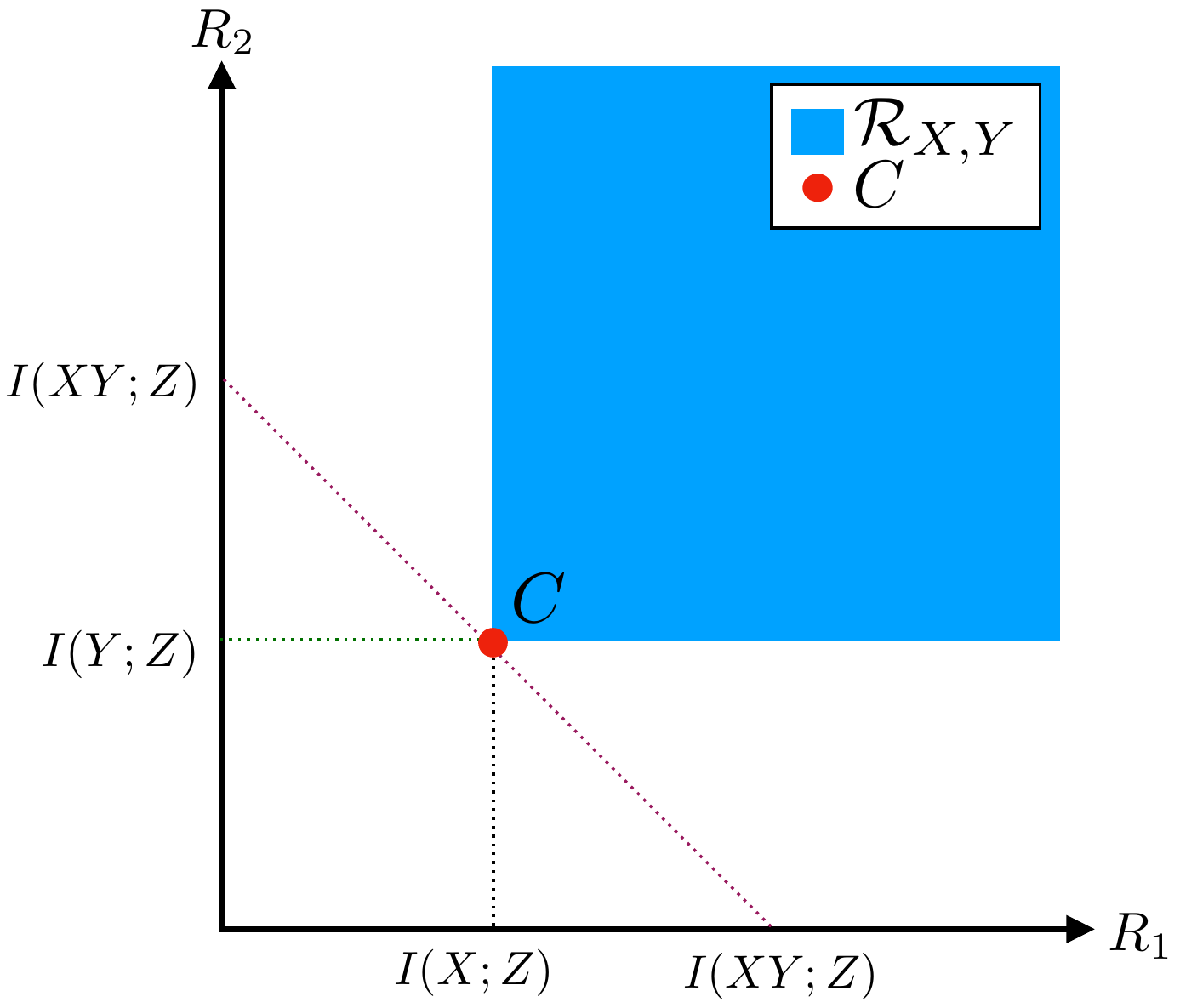}
  \caption{Region $\mathcal{R}_{X,Y}$ in  Case 2: $I(XY;Z)=~ I(X;Z)+ I(Y;Z)$.}
  \label{fig:model2}
  \end{minipage}
\end{figure}
  \subsection{Encoding Scheme for Case 1} \label{sec:sub1}
Consider the region  $\mathcal{R}_{X,Y}$ for some product distribution $p_Xp_Y$ on $\mathcal{X} \times \mathcal{Y}$ such that $I(XY;Z)>I(X;Z)+I(Y;Z)$. Since $\mathcal{R}_{X,Y}$ is a contrapolymatroid \cite{edmonds2003submodular}, to achieve the region  $\mathcal{R}_{X,Y}$, it is sufficient to achieve any rate pair $(R_1,R_2)$ of the dominant face $\mathcal{D}$ of $\mathcal{R}_{X,Y}$. 
        We next show that $\mathcal{D}$ can be achieved through rate-splitting using the following lemma proved in Appendix \ref{AppendixA}.
\begin{lem} \label{lem1}
Consider $f:~ \mathcal{Y} \times~ \mathcal{Y} \rightarrow \mathcal{Y}$, $(u,v)\mapsto \max (u,v)$, and form  $(\mathcal{Y}\times\mathcal{Y},p_{U_\epsilon}p_{V_\epsilon}), \epsilon \in [0,1]$, such that $p_{U_\epsilon V_\epsilon}=p_{U_\epsilon}p_{V_\epsilon}$, $p_{f(U_\epsilon,V_\epsilon)}=p_Y$, for fixed $(y, u), p_{f(U_{\epsilon},V_{\epsilon})|U_\epsilon}(y| u) $ is a continuous function of $\epsilon$, and
\begin{align}
& U_{\epsilon=0}=0=V_{\epsilon=1},\label{lem1eqn1}\\
& U_{\epsilon=1}=f(U_{\epsilon=1},V_{\epsilon=1}),\label{lem1eqn2}\\
& V_{\epsilon=0}=f(U_{\epsilon=0},V_{\epsilon=0}).\label{lem1eqn3}
\end{align}
The above construction is indeed possible as shown in {\cite [Example 3]{grant2001rate}}.
 Then, we have $I(XY;Z)=R_1+R_U+R_V$, where we have defined the functions
\begin{align*}
   & R_1: [0, 1]\to\ \mathbb{R}^+,\epsilon \mapsto I(X;Z|U_\epsilon),\\
   &  R_U:[0, 1]\to\ \mathbb{R}^+,\epsilon \mapsto I(U_\epsilon;Z),\\
   & R_V:[0, 1]\to\ \mathbb{R}^+,\epsilon \mapsto I(V_\epsilon;Z|U_\epsilon X).
  \end{align*}
Moreover, $R_1$ is continuous with respect to $\epsilon$ and $[ I(X;Z), I(X;Z|Y)]$ is contained in its image. 

When the context is clear, we do not explicitly write the dependence of $U$ and $V$ with respect to $\epsilon$ by dropping the subscript $\epsilon$.
\end{lem}

Fix a point $(R_1,R_2)$ in $\mathcal{D}$. By Lemma \ref{lem1}, there exists a joint probability distribution $q_{UVXYZ}$ over $\mathcal{Y}\times \mathcal{Y}\times \mathcal{X}\times \mathcal{Y}\times \mathcal{Z}$ such that $R_1= I(X;Z|U)$, $R_2 = R_U +R_V$ with  $R_U =I(U;Z)$ and $R_V = I(V;Z|UX)$. We provide next a coding scheme that will be shown to achieve the point $(R_1,R_2)$.  \textcolor{black}{The encoding scheme operates over $k \in \mathbb{N}$ blocks of length $N$ and is described in Algorithms~$\ref{alg:encoding_11}$ and~$\ref{alg:encoding_21}$. A high level description of the encoding scheme is as follows. For the first transmitter, we perform source resolvability for the discrete memoryless source $(\mathcal{X}, q_X)$ using randomness with rate $H(X)$ in Block~$1$. Using Lemma $\ref{lem1}$, we perform rate splitting for the second transmitter to get two virtual users such that one virtual user is associated with 
the discrete memoryless source  $(\mathcal{Y}, q_U)$ and the other virtual user is associated with the discrete memoryless source $(\mathcal{Y}, q_V)$. Then, we perform source resolvability with rates $H(U)$ and $H(V)$ for the discrete memoryless sources$(\mathcal{Y}, q_U)$ and $(\mathcal{Y}, q_V)$, respectively.
For the next encoding blocks, we proceed as in Block  $1$ using source resolvability and rate splitting except that part of the randomness is now recycled from the previous block. More precisely, we recycle the bits of randomness used at the inputs of the channel in the previous block that are almost independent from the channel output. The rates of those bits will be shown to approach  $H(X|UZ)$, $H(U|Z)$, $H(V|UZX)$ for User $1$ and the two virtual users,~respectively. }
\begin{itemize}
    \item The encoding at Transmitter $1$ is described in Algorithm~\ref{alg:encoding_11} and uses
    \begin{itemize}
    \item A hash function $G_X:\{0,1\}^{N} \longrightarrow \{0,1\}^{r_X}$ chosen uniformly at random in a family of two-universal hash functions, where the output length of the hash function $G_X$ is defined as follows 
    \begin{align}
        {r_X}\triangleq N(H(X|UZ)- \epsilon_1 /2),\label{eqnrX}
    \end{align}
    where $\epsilon_1 \triangleq 2( \delta_{\mathcal{A}}(N)+\xi)$, $\delta_{\mathcal{A}}(N)\triangleq \log (|\mathcal{Y}|^2|\mathcal{X}|+3)\sqrt{\frac{2}{N}(3+\log N)}$, $\xi>0$.
\item A source resolvability code  for the discrete memoryless source $(\mathcal{X},q_{X})$ with encoder function $e_N^X$ and rate $H(X)+\frac{\epsilon_1 }{2}$, such that the distribution of the encoder output $\widetilde{p}_{X^{1:N}}$ satisfies $\mathbb{V}(\widetilde{p}_{X^{1:N}},q_{X^{1:N}})\leq \delta(N)$,  where $\delta(N)$ is such that $\lim_{N \to +\infty} \delta(N) =0$.
\end{itemize}
In Algorithm~\ref{alg:encoding_11}, the hash function output $\widetilde{E}_i$, $i\in \llbracket 2,k\rrbracket$, with length $r_X$ corresponds to recycled randomness from Block $i-1$. 

\newcommand{\floor}[1]{\left\lfloor #1 \right\rfloor}

\item The encoding at Transmitter $2$ is described in Algorithm~\ref{alg:encoding_21} and uses
\begin{itemize}
    \item Two hash functions $G_U:\{0,1\}^{N} \longrightarrow \{0,1\}^{r_U}$ and $G_V:\{0,1\}^{N} \longrightarrow \{0,1\}^{r_V}$ chosen uniformly at random in families of two-universal hash functions, where  the output lengths of the hash functions  $G_U$ and $G_V$ are defined as follows
\begin{align}
   \nonumber \quad r_U & \triangleq N(H(U|Z)- \epsilon_1 / 2),\\  \quad r_V & \triangleq N(H(V|UZX)- \epsilon_1 /2) \label{def1}.
\end{align}
    \item A source resolvability code for the discrete memoryless source  $(\mathcal{U},q_{U})$ with encoding function $e_N^U$ and rate $H(U)+\frac{\epsilon_1 }{2}$,  such that the distribution of the encoder output $\widetilde{p}_{U^{1:N}}$ satisfies $\mathbb{V}(\widetilde{p}_{U^{1:N}},q_{U^{1:N}})\leq \delta(N)$, where $\delta(N)$ is such that $\lim_{N \to +\infty} \delta(N) =0$. 
\item A source resolvability code for the discrete memoryless source  $(\mathcal{V},q_{V})$ with encoding function $e_N^V$ and rate $H(V)+\frac{\epsilon_1 }{2}$, such that  the distribution of the encoder output $\widetilde{p}_{V^{1:N}}$ satisfies $\mathbb{V}(\widetilde{p}_{V^{1:N}},q_{V^{1:N}})\leq \delta(N)$, where $\delta(N)$ is such that $\lim_{N \to +\infty} \delta(N) =0$. \end{itemize}
In Algorithm~\ref{alg:encoding_21}, the hash function outputs $\widetilde{D}_i$ and $\widetilde{F}_i$, $i\in \llbracket 2,k\rrbracket$, with lengths $r_U$ and $r_V$, respectively, correspond to recycled randomness from Block $i-1$. 

\end{itemize} 
The dependencies between the random variables involved in Algorithms \ref{alg:encoding_11} and \ref{alg:encoding_21} are represented in Figure~\ref{figFGD}.

 \begin{algorithm}[h]
  \caption{Encoding algorithm at Transmitter $1$ in Case 1}
  \label{alg:encoding_11}
  \begin{algorithmic}   [1] 
     \REQUIRE A vector $ E_1$ of $N (H(X)+\epsilon_1)$ uniformly distributed bits,   and for $i \in \llbracket 2,k \rrbracket$, a vector $E_{i}$ of $N (I(X; UZ)+\epsilon_{1})$ uniformly distributed bits. 
                                 \FOR{Block $i=1$ to $k$}
                         \IF{$i=1$}
                         \STATE Define $\widetilde X_1^{1:N}\triangleq e^X_N(E_1)$ 
                         \ELSIF{$i>1$}
      \STATE  Define $\widetilde E_i \triangleq G_X (\widetilde X_{i-1}^{1:N})$
       \STATE Define $\widetilde X_i^{1:N}\triangleq e^X_N(\widetilde E_i \lVert E_i)$, where $\lVert$ denotes concatenation
                                  \ENDIF
                                  \STATE Send $\widetilde X_i^{1:N}$  over the channel
                        \ENDFOR
       \end{algorithmic}
\end{algorithm}
\begin{algorithm}[h]
  \caption{Encoding algorithm at Transmitter  $2$ in Case 1}
  \label{alg:encoding_21}
  \begin{algorithmic}   [1] 
    \REQUIRE A vector $ D_1$ of $N (H(U)+\epsilon_1)$ uniformly distributed bits,   and for $i \in \llbracket 2,k \rrbracket$, a vector $D_{i}$ of $N (I(U; Z)+\epsilon_{1})$ uniformly distributed bits. A vector $ F_1$ of $N (H(V)+\epsilon_1)$ uniformly distributed bits, and for $i \in \llbracket 2,k \rrbracket$, a vector $F_{i}$ of $N (I(V; UZX)+\epsilon_{1})$ uniformly distributed bits.
                                 \FOR{Block $i=1$ to $k$}
                         \IF{$i=1$}
                         \STATE Define $\widetilde U_1^{1:N}\triangleq e^U_N(D_1)$
                        and $\widetilde V_1^{1:N}\triangleq e^V_N(F_1)$
                         \ELSIF{$i>1$}
      \STATE  Define $\widetilde D_i \triangleq G_U (\widetilde U_{i-1}^{1:N})$ and $\widetilde F_i \triangleq G_V (\widetilde V_{i-1}^{1:N})$
       \STATE Define $\widetilde U_i^{1:N}\triangleq e^U_N(\widetilde D_i \lVert D_i)$
         and $\widetilde V_i^{1:N}\triangleq e^V_N(\widetilde F_i \lVert F_i)$ \STATE \vspace*{-1em}Define $\widetilde Y_i^{1:N}\triangleq f(\widetilde U_i^{1:N},\widetilde V_i^{1:N})$, where $f$ is defined in Lemma \ref{lem1}
                                  \ENDIF
                                  \STATE Send $\widetilde Y_i^{1:N}$  over the channel
                        \ENDFOR
        \end{algorithmic}
\end{algorithm}
\begin{figure}
\centering
  \includegraphics[width=8.5 cm]{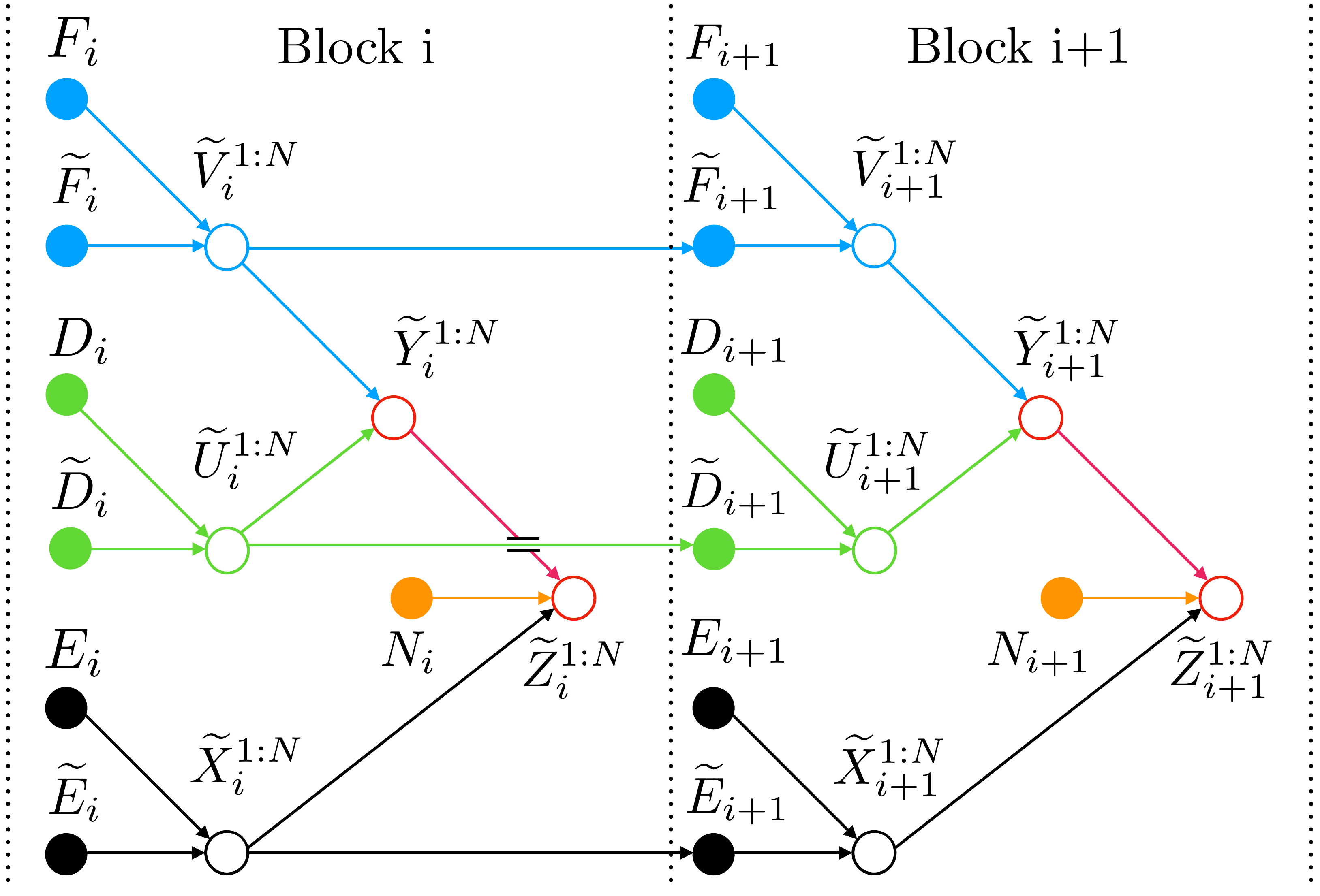}
  \caption{Dependence graph for the random variables involved in the encoding for Case 1. $N_i$, $i \in \llbracket 1 ,k \rrbracket$, is the channel noise corresponding to the transmission over Block $i$.  For Block $i\in \llbracket 2, k \rrbracket$, $(D_i, \widetilde{D}_{i}$), $(F_i, \widetilde{F}_{i}$), $(E_i, \widetilde{E}_{i}$) are the random sequences used at the encoders to form $\widetilde{U}^{1:N}_i$, $\widetilde{V}^{1:N}_i$, $\widetilde{X}^{1:N}_i$, respectively.
  }
  \label{figFGD}
\end{figure}
\subsection{Encoding Scheme for Case 2} \label{sec:sub2}\textcolor{black}{
The encoding scheme for Case $2$ is same as the encoding for Case $1$ with the substitutions $ U\leftarrow \emptyset$ and $V\leftarrow Y$. }
\section{Coding Scheme Analysis}\label{sec:coding scheme analysis}
\subsection{Coding Scheme Analysis for Case $1$}

 \textcolor{black}{First, we show that in each encoding Block $i\in \llbracket 1, k \rrbracket$, the random variables $\widetilde U_i^{1:N}, \widetilde V_i^{1:N}, \widetilde X_i^{1:N}, \widetilde Y_i^{1:N}, \widetilde Z_i^{1:N}$ 
 induced by the coding scheme approximate well the target distribution $q_{U^{1:N}V^{1:N}X^{1:N}Y^{1:N}Z^{1:N}}$. Then, we  show  that  the  target  output  distribution $q_{Z^{1:kN}}$  is  well approximated jointly over all blocks. To do so, we show that the recycled randomness $\widetilde E_i, \widetilde D_i, \widetilde F_i$ in Block $i\in \llbracket 2, k \rrbracket$ that appears in Line $5$ of Algorithms~\ref{alg:encoding_11} and \ref{alg:encoding_21} is almost independent of the channel output in Block $i-1$. Note that randomness recycling is studied via a distributed version of the leftover hash lemma stated in Lemma~$\ref{lem11}$.
 Finally, we prove that the encoding scheme of Section \ref{sec:sub1} achieves the desired rate-tuple.}
 
       For convenience, define $\widetilde E_{1} \triangleq \emptyset$, $\widetilde D_{1}\triangleq \emptyset$, and $\widetilde F_{1}\triangleq \emptyset$. Let 
       \begin{align}
           \widetilde p_{ E_{i} D_{i} F_{i} X^{1:N}_{i} U^{1:N}_{i} V^{1:N}_{i} Y^{1:N}_{i} Z^{1:N}_{i}} \label{eqnP}
       \end{align}
       denote the joint probability distribution of the random variables $\widetilde E_{i}, \widetilde D_{i}, \widetilde F_{i}, \widetilde X^{1:N}_{i}, \widetilde U^{1:N}_{i}, \widetilde V^{1:N}_{i}, \widetilde Y^{1:N}_{i}$, and $\widetilde Z^{1:N}_{i}$ created in Block $i \in \llbracket 1,k\rrbracket$ of the coding scheme of Section \ref{sec:sub1}.

 We first prove in the following lemma that in Block $i\in \llbracket 2,k \rrbracket$, if the inputs $\widetilde X_{i-1}^{1:N}$, $\widetilde U_{i-1}^{1:N}$, $\widetilde V_{i-1}^{1:N}$ of the hash functions $G_X$, $G_U$, $G_V$, respectively, are replaced by $X^{1:N}$, $U^{1:N}$, $ V^{1:N}$ distributed according to $q_{X^{1:N}U^{1:N}V^{1:N}} \triangleq \prod_{i=1}^N q_{X U V}$, then the output of these hash functions are almost jointly uniformly~distributed. 
    \begin{lem} \label{lem3} 
  Let $ p^{unif}_{\bar E}, p^{unif}_{\bar D}, p^{unif}_{\bar F}$ denote the uniform distributions over $\{0,1\}^{r_X}$, $\{0,1\}^{r_U}$, $\{0,1\}^{r_V}$, respectively. 
  Then, 
       \begin{align*}
           &\mathbb{V} \left(q_{G_{X}({X}^{1:N})G_{U}({U}^{1:N})G_{V}({V}^{1:N})Z^{1:N}}, p^{unif}_{\bar E}p^{unif}_{\bar D}p^{unif}_{\bar F}q_{Z^{1:N}} \right)\\ & \phantom{--}   \leq     \textcolor{black}{\delta^{(0)}(N)},
     \end{align*}
where 
$ \textcolor{black}{\delta^{(0)}(N)}
 \triangleq 2/N + \sqrt{7} \cdot 2^{-\frac{N\xi}{2}}.$
\end{lem}                 \begin{proof}      
Define $\mathcal{A} \triangleq \{U,V,X\}$ and, for any $ \mathcal{S}\subseteq \mathcal{A}$, define 
 $T^{}_{\mathcal{S}}\triangleq ( W^{})_{W \in \mathcal{S}}.$ Hence,      we have 
  \begin{align*}
     T_{\mathcal{A}}^{1:N}&= ({{X}}^{1:N},{{U}}^{1:N},{{V}}^{1:N}),\\
     q_{T_{\mathcal{A}}^{1:N}Z^{1:N}}&=q_{ X^{1:N}  U^{1:N} V^{1:N}  Z^{1:N}}. 
 \end{align*}
Then, by Lemma \ref{lems1} in Appendix~\ref{appA}, applied to the product distribution $q_{T_{\mathcal{A}}^{1:N}Z^{1:N}}$, there exists a subnormalized non-negative function $w_{{T_{\mathcal{A}}^{1:N}}Z^{1:N}}$ such that, for any $\mathcal{S}\subseteq \mathcal{A}$,
\begin{align}
         \mathbb{V}({w}_{ X^{1:N}  U^{1:N}  V^{1:N}  Z^{1:N}}, q_{ X^{1:N}  U^{1:N} V^{1:N}  Z^{1:N}}) \leq   1/N ,\label{lem3eqn3}\\ 
           H_{\infty}({w}_{{{T}_{\mathcal{S}}^{1:N}} Z^{1:N}}|q_{ Z^{1:N}})\geq\!\! N H({{ T}_{\mathcal{S}}}| Z) -N \delta_{\mathcal{S}}(N),\label{lem3eqn4}
       \end{align} \textcolor{black}{where the min-entropy $H_{\infty}({w}_{{{T}_{\mathcal{S}}^{1:N}} Z^{1:N}}|q_{ Z^{1:N}})$ is defined in Lemma \ref{lems1} in Appendix~\ref{appA}, and $\delta_{\mathcal{S}}(N)\triangleq \log (\lvert\mathcal{T}_{\mathcal{S}}\rvert+3) \sqrt{\frac{2}{N}(3+\log N)}$  with $ \mathcal{T}^{}_{\mathcal{S}}$ is the domain over which ${T}_{\mathcal{S}}$ is defined.} Next, let $q_{EDF}$ define the joint distribution of 
          \begin{align}
              E \triangleq G_X ( X^{1:N}), D \triangleq G_U ( U^{1:N}), F \triangleq G_V ( V^{1:N}),\label{eqnedf}
          \end{align}
          where $U^{1:N}$, $V^{1:N}$, and $X^{1:N}$ are distributed according to $q_{U^{1:N}V^{1:N}X^{1:N}}$. Then, we have
               \begin{align*}\nonumber
      &  \mathbb{V}(q_{EDFZ^{1:N}},  p^{unif}_{\bar E}p^{unif}_{\bar D}p^{unif}_{\bar F}q_{Z^{1:N}})  \\ \nonumber 
      &\stackrel{(a)}\leq \mathbb{V}(q_{EDFZ^{1:N}}, w_{EDF Z^{1:N}})\\\nonumber & \phantom{--}+   \mathbb{V}(w_{EDFZ^{1:N}},  p^{unif}_{\bar E}p^{unif}_{\bar D}p^{unif}_{\bar F}q_{Z^{1:N}})\\ \nonumber
       &\stackrel{(b)}= \mathbb{V}(q_{G_{X}({X}^{1:N})G_{U}({U}^{1:N})G_{V}({V}^{1:N})Z^{1:N}},\\ \nonumber & \phantom{----} w_{G_{X}({X}^{1:N})G_{U}({U}^{1:N})G_{V}({V}^{1:N}) Z^{1:N}})  \\\nonumber & \phantom{--} +\mathbb{V}(w_{EDFZ^{1:N}},  p^{unif}_{\bar E}p^{unif}_{\bar D}p^{unif}_{\bar F}q_{Z^{1:N}})  \\ \nonumber
       & \stackrel{(c)} \leq \mathbb{V}( q_{ X^{1:N}  U^{1:N} V^{1:N}Z^{1:N}},{w}_{ X^{1:N}  U^{1:N}  V^{1:N} Z^{1:N} })\\\nonumber & \phantom{--}+\mathbb{V}(w_{EDFZ^{1:N}},  p^{unif}_{\bar E}p^{unif}_{\bar D}p^{unif}_{\bar F}q_{Z^{1:N}}) \\ \nonumber
                 & \stackrel{(d)} \leq 1/N +   \mathbb{V}(w_{EDF Z^{1:N}},  p^{unif}_{\bar E}p^{unif}_{\bar D}p^{unif}_{\bar F} w_{Z^{1:N}})\\\nonumber & \phantom{--}+ \mathbb{V}( p^{unif}_{\bar E}p^{unif}_{\bar D}p^{unif}_{\bar F} w_{Z^{1:N}},  p^{unif}_{\bar E}p^{unif}_{\bar D}p^{unif}_{\bar F} q_{Z^{1:N}}) \\ \nonumber
           & \stackrel{(e)} \leq 2/N +   \mathbb{V}(w_{EDF Z^{1:N}},  p^{unif}_{\bar E}p^{unif}_{\bar D}p^{unif}_{\bar F} w_{Z^{1:N}}) \\ \nonumber
      & \stackrel{(f)}  \leq 2/N +  \sqrt{{ \displaystyle\sum_{\substack{{\mathcal S\subseteq\mathcal A},  {\mathcal S \neq \emptyset}}}}2^{r_{\mathcal S}-H_{\infty}({w}_{{{T}_{\mathcal{S}}^{1:N}} Z^{1:N}}|q_{ Z^{1:N}})}}\\\nonumber
            &\stackrel{(g)}\leq 2/N + {{{\sqrt{{ \displaystyle\sum_{\substack{{\mathcal S\subseteq\mathcal A}, {\mathcal S \neq \emptyset}}}}2^{r_{\mathcal S}-NH( T_{\mathcal S}|Z)+N \delta_{\mathcal{S}}(N)}}}}}\\ 
     &\stackrel{(h)}\leq 2/N + {{{\sqrt{{ \displaystyle\sum_{\substack{{\mathcal S\subseteq\mathcal A}, {\mathcal S \neq \emptyset}}}}2^{r_{\mathcal S}-NH( T_{\mathcal S}|Z)+N \delta_{\mathcal{A}}(N)}}}}} 
         \end{align*}
    where $(a)$ holds by the triangle inequality, $(b)$ holds by \eqref{eqnedf}, $(c)$ holds by the data processing inequality, $(d)$  holds by (\ref{lem3eqn3}) and the triangle inequality,   $(e)$ holds by (\ref{lem3eqn3}), $(f)$ holds by Lemma \ref{lem11} in Appendix \ref{appA} \textcolor{black}{and $r_{\mathcal{S}}\triangleq \sum_{i \in \mathcal{S} }r_i$ similar to the notation of Lemma \ref{lem11}}, $(g)$ holds by  (\ref{lem3eqn4}), $(h)$ holds because for any $\mathcal{S}\subseteq \mathcal{A}$, $\delta_{\mathcal{S}}(N) \leq \delta_{\mathcal{A}}(N)$. Next, we have
           \begin{align}
                &{{{\sqrt{{ \displaystyle\sum_{\substack{{\mathcal S\subseteq\mathcal A}, {\mathcal S \neq \emptyset}}}}2^{r_{\mathcal S}-NH( T_{\mathcal S}|Z)+N \delta_{\mathcal{A}}(N)}}}}}   \nonumber \\ \nonumber
      &\stackrel{(a)} = \left({{{{2^{N(H(X|UZ)-\frac{\epsilon_1}{2})-NH( X|Z)}}}}}\right.\\\nonumber & \phantom{--}+ 2^{N(H(U|Z)-\frac{\epsilon_1}{2})-N H( U|Z)}  \\ \nonumber  
      & \phantom{--}+2^{N(H(V|UZX)-\frac{\epsilon_1}{2})-N H(V|Z)}  \\ \nonumber  
      & \phantom{--}+ 2^{N(H(X|UZ)-\frac{\epsilon_1}{2})+N(H(U|Z)-\frac{\epsilon_1}{2})-N H(XU|Z)} \\\nonumber
      & \phantom{--} + 2^{N(H(U|Z)-\frac{\epsilon_1}{2})+N(H(V|UZX)-\frac{\epsilon_1}{2})-NH(UV|Z)}\\\nonumber 
     & \phantom{--} +   2^{N(H(V|UZX)-\frac{\epsilon_1}{2})+N(H(X|UZ)-\frac{\epsilon_1}{2})-N H(VX|Z)}  \\ \nonumber  
   &\phantom{--} +2^{N(H(X|UZ)-\frac{\epsilon_1}{2})+N(H(U|Z)-\frac{\epsilon_1}{2})+N(H(V|UZX)-\frac{\epsilon_1}{2})}\  \\ \nonumber  
      & \phantom{---} \left. \times 2^{-N H(XUV|Z)}\right)^{\frac{1}{2}}   \times 2^{\frac{1}{2}N \delta_{\mathcal{A}}(N)}\\\nonumber
   & \stackrel {(b)}= \left({{{{2^{-NI(X;U|Z)-N\frac{\epsilon_1}{2}}}}}}+ 2^{-N\frac{\epsilon_1}{2}}+  2^{-NI(V;UX|Z)-N\frac{\epsilon_1}{2}}\right.\\\nonumber
     & \phantom{--} +2^{-N{\epsilon_1}} +2^{-N{\epsilon_1}-NI(V;X|UZ)} +2^{-N\frac{3\epsilon_1}{2}} \\ \nonumber  
      & \phantom{--} \left. +2^{-NI(V;U|ZX)-NI(X;U|Z)-N{\epsilon_1}} \right)^{\frac{1}{2}} \times 2^{\frac{1}{2}N \delta_{\mathcal{A}}(N)}\\\nonumber
    & \stackrel {(c)}\leq \delta^{(0)}(N) - 2/N \xrightarrow{N \to {+}\infty}0,
           \end{align}
           where $(a)$ holds by $(\ref{eqnrX})$ and $(\ref{def1})$, $(b)$ holds by the definition of mutual information and the chain rule for entropy, $(c)$ holds by the definition of $\delta^{(0)}(N)$ and because $\epsilon_1= 2( \delta_{\mathcal{A}}(N)+\xi)$. 
    \end{proof}
        We now show that in each encoding block, the random variables induced by the coding scheme approximate well the target distribution.
          \begin{lem} \label{lemblock}
           For Block $i\in \llbracket 1, k \rrbracket$, we have
 \begin{align*}
&\mathbb{V}(\widetilde p_{U_i^{1:N}V_i^{1:N}X_i^{1:N}Y_i^{1:N}Z_i^{1:N}}, q_{U^{1:N}V^{1:N}X^{1:N}Y^{1:N}Z^{1:N}})\leq \!\delta_i(N),
 \end{align*}
where $\delta_i(N) \triangleq \frac{3}{2}(\delta(N) + \textcolor{black}{\delta^{(0)}(N)}) ({3^{i}-1}{}) + 3^{i+1} \delta(N)$.
 \end{lem}
  \begin{proof} 
  We prove the result by induction. We first prove that the lemma holds for $i=1$. Remark that
         \begin{align}
 \widetilde p_{Y^{1:N}_1 \vert U_1^{1:N}V_1^{1:N}X_1^{1:N}}\nonumber
 &\stackrel{(a)}=\widetilde p_{Y^{1:N}_1 \vert U_1^{1:N}V_1^{1:N}}\\\nonumber
 &\stackrel{(b)}=q_{Y^{1:N} \vert U^{1:N}V^{1:N}}\\\
 &\stackrel{(c)}=q_{Y^{1:N} \vert U^{1:N}V^{1:N}X^{1:N}}\label{lem4eqn3},
 \end{align}
  where (a) holds because $\widetilde X_1^{1:N}$  is independent from $(\widetilde U^{1:N}_1,\widetilde V^{1:N}_1,\widetilde Y^{1:N}_1)$, (b) holds by the construction of $Y^{1:N}$ and  $\widetilde Y_1^{1:N}$, (c) holds because  $X^{1:N}$ is independent from $(U^{1:N},V^{1:N},Y^{1:N})$. Next, we have
\begin{align}\label{mideqn}
& \mathbb{V}(\widetilde p_{U_1^{1:N}V_1^{1:N}X_1^{1:N}Y_1^{1:N}Z_1^{1:N}}, q_{U^{1:N}V^{1:N}X^{1:N}Y^{1:N}Z^{1:N}}) \nonumber \\\nonumber
&\stackrel{(a)}=\mathbb {V}( \widetilde p_{Z_1^{1:N}|X_1^{1:N}Y_1^{1:N}}\widetilde p_{U_1^{1:N}V_1^{1:N}X_1^{1:N}Y_1^{1:N}},\\ \nonumber  
      & \phantom{--} q_{Z^{1:N}|X^{1:N}Y^{1:N}} q_{U^{1:N}V^{1:N}X^{1:N}Y^{1:N}})\\\nonumber
&\stackrel{(b)}=\mathbb {V}(\widetilde p_{U_1^{1:N}V_1^{1:N}X_1^{1:N}Y_1^{1:N}}, q_{U^{1:N}V^{1:N}X^{1:N}Y^{1:N}}) \nonumber \\\nonumber
&\stackrel{(c)}=\mathbb{V}(\widetilde p_{U_1^{1:N}V_1^{1:N}X_1^{1:N}}, q_{U^{1:N}V^{1:N}X^{1:N}})
\\\nonumber
&\stackrel{(d)}=\mathbb{V}(\widetilde p_{X_1^{1:N}}\widetilde p_{U_1^{1:N}V_1^{1:N}}, q_{X^{1:N}}q_{U^{1:N}V^{1:N}})\\\nonumber
&\stackrel{(e)}\leq \mathbb{V}(\widetilde p_{X_1^{1:N}}\widetilde p_{U_1^{1:N}V_1^{1:N}}, q_{X^{1:N}}\widetilde p_{U_1^{1:N}V_1^{1:N}})\\ \nonumber  
      & \phantom{--}+ \mathbb{V}(q_{X^{1:N}}\widetilde p_{U_1^{1:N}V_1^{1:N}}, q_{X^{1:N}}q_{U^{1:N}V^{1:N}})\\\nonumber
&\stackrel{(f)}=\mathbb{V}(\widetilde p_{X_1^{1:N}}, q_{X^{1:N}})+\mathbb{V}(\widetilde p_{U_1^{1:N}}\widetilde p_{V_1^{1:N}}, q_{U^{1:N}}q_{V^{1:N}})\\\nonumber
&\stackrel{(g)}\leq \mathbb{V}(\widetilde p_{X_1^{1:N}}, q_{X^{1:N}})+\mathbb{V}(\widetilde p_{U_1^{1:N}}\widetilde p_{V_1^{1:N}}, q_{U^{1:N}}\widetilde p_{V_1^{1:N}})\\ \nonumber  
      & \phantom{--}+\mathbb{V}(q_{U^{1:N}}\widetilde p_{V_1^{1:N}}, q_{U^{1:N}}q_{V^{1:N}})\\
&=\mathbb{V}(\widetilde p_{X_1^{1:N}}, q_{X^{1:N}})+\mathbb{V}(\widetilde p_{U_1^{1:N}}, q_{U^{1:N}})+\mathbb{V}(\widetilde p_{V_1^{1:N}}, q_{V^{1:N}})\\\nonumber
&\stackrel{(h)}\leq 3 \delta(N), 
\end{align}
where $(a)$ holds by  the two Markov chains $(U^{1:N},V^{1:N})-(X^{1:N},Y^{1:N})- Z^{1:N}$ and $(\widetilde U_1^{1:N},\widetilde V_1^{1:N})-(\widetilde X^{1:N}_1,\widetilde Y_1^{1:N}) -\widetilde Z_1^{1:N}$, $(b)$ holds because $q_{Z^{1:N}\vert X^{1:N}Y^{1:N} } = \widetilde p_{Z_1^{1:N}\lvert X_1^{1:N}Y_1^{1:N}}$, $(c)$ holds by (\ref{lem4eqn3}), $(d)$ holds because $ X^{1:N}$ is independent from $(U^{1:N}, V^{1:N})$ and $\widetilde X_1^{1:N}$ is independent from $(\widetilde U_1^{1:N}, \widetilde V_1^{1:N})$, $(e)$  holds by the  triangle inequality,  $(f)$ holds because $ U^{1:N}$ is independent from  $V^{1:N}$ and $\widetilde U_1^{1:N}$ is independent from $\widetilde V_1^{1:N}$, $(g)$  holds by the  triangle inequality, $(h)$ holds by the source resolvability codes used at the transmitters because $\frac{|E_1|}{N}>H(X) + \epsilon_1/2, \frac{ | D_1|}{N}>H(U)+ \epsilon_1/2, \frac{|F_1|}{N}>H(V)+ \epsilon_1/2$.

Assume now that, for $i \in \llbracket 2,  k-1 \rrbracket$,
 the lemma holds. For $i\in \llbracket 2,k \rrbracket$, consider $\bar E_{i}, \bar D_{i}, \bar F_{i}$ distributed according to $ {p^{unif}_{\bar E}}, {p^{unif}_{\bar D}}, {p^{unif}_{\bar F}}$, respectively.  Let $p_{\bar X_{i}^{1:N}}, p_{\bar U_{i}^{1:N}}, p_{\bar V_{i}^{1:N}}$ denote the distribution of $\bar X_{i}^{1:N}\triangleq e^X_N(\bar E_{i}, E_{i}), \bar U_{i}^{1:N}\triangleq e^U_N(\bar D_{i}, D_{i}), \bar V_{i}^{1:N}\triangleq e^V_N(\bar F_{i}, F_{i})$, respectively. Then, for $i \in~\llbracket 1, k-1 \rrbracket$, we have
       \begin{align}\nonumber
&\mathbb{V}(\widetilde p_{U_{i+1}^{1:N}V_{i+1}^{1:N}X_{i+1}^{1:N}Y_{i+1}^{1:N}Z_{i+1}^{1:N}}, q_{U^{1:N}V^{1:N}X^{1:N}Y^{1:N}Z^{1:N}})\\\nonumber
&\stackrel{(a)} \leq  \mathbb{V}(\widetilde p_{X_{i+1}^{1:N}}, q_{X^{1:N}})+\mathbb{V}(\widetilde p_{U_{i+1}^{1:N}}, q_{U^{1:N}})+\mathbb{V}(\widetilde p_{V_{i+1}^{1:N}}, q_{V^{1:N}})\\\nonumber
&\stackrel{(b)} \leq \mathbb{V}(\widetilde p_{X_{i+1}^{1:N}},  p_{\bar X_{i+1}^{1:N}})+\mathbb{V}( p_{\bar X_{i+1}^{1:N}}, q_{X^{1:N}})\\ \nonumber  
      & \phantom{--}+\mathbb{V}(\widetilde p_{U_{i+1}^{1:N}},  p_{\bar U_{i+1}^{1:N}}) +\mathbb{V}(p_{\bar U_{i+1}^{1:N}}, q_{U^{1:N}})\\ \nonumber  
      & \phantom{--}+\mathbb{V}(\widetilde p_{V_{i+1}^{1:N}},  p_{\bar V_{i+1}^{1:N}})+ \mathbb{V}( p_{\bar V_{i+1}^{1:N}}, q_{V^{1:N}})\\\nonumber
&\stackrel{(c)}\leq  3 \delta(N)+\mathbb{V}(\widetilde p_{X_{i+1}^{1:N}},  p_{\bar X_{i+1}^{1:N}})+\mathbb{V}(\widetilde p_{U_{i+1}^{1:N}},  p_{\bar U_{i+1}^{1:N}})\\ \nonumber  
      & \phantom{--}+\mathbb{V}(\widetilde p_{V_{i+1}^{1:N}},  p_{\bar V_{i+1}^{1:N}})\\\nonumber 
&\stackrel{(d)}\leq  3 \delta(N)+\mathbb{V}(\widetilde p_{E_{i+1}},  p^{unif}_{\bar E})+\mathbb{V}(\widetilde p_{D_{i+1}},  p^{unif}_{\bar D})\\  
      & \phantom{--}+\mathbb{V}(\widetilde p_{F_{i+1}},  p^{unif}_{\bar F})\label{lem4eqn7},
\end{align}
where $(a)$ holds similar to (\ref{mideqn}), $(b)$ holds by the triangle inequality, $(c)$ holds by the source resolvability codes used at the transmitters because $\frac{|\bar E_i|+|E_i|}{N} = H(X) + \epsilon_1/2, \frac{ |\bar F_i|+|F_i|}{N}=H(V)+ \epsilon_1/2, \frac{|\bar D_i|+|D_i|}{N} = H(U)+ \epsilon_1/2$, $(d)$ holds by the data processing inequality. Next, we have 
\begin{align}\nonumber
    & \max \left(\! \mathbb{V}(\widetilde p_{E_{i+1}},  p^{unif}_{\bar E}),\mathbb{V}(\widetilde p_{D_{i+1}},  p^{unif}_{\bar D}),\mathbb{V}(\widetilde p_{F_{i+1}},  p^{unif}_{\bar F}) \!\right) \\ \nonumber
     & \leq \mathbb{V}(\widetilde p_{E_{i+1}D_{i+1}F_{i+1}},  p^{unif}_{\bar E} p^{unif}_{\bar D} p^{unif}_{\bar F}
     )\\\nonumber
    & \stackrel{(a)}\leq \mathbb{V}(\widetilde p_{E_{i+1}D_{i+1}F_{i+1}},  q_{G_{X}({X}^{1:N})G_{U}({U}^{1:N})G_{V}({V}^{1:N})})\\\nonumber
    & \phantom{--}+\mathbb{V}(q_{G_{X}({X}^{1:N})G_{U}({U}^{1:N})G_{V}({V}^{1:N})},  p^{unif}_{\bar E} p^{unif}_{\bar D} p^{unif}_{\bar F} 
    )\\ 
   \nonumber  & \stackrel{(b)}= \mathbb{V}(\widetilde p_{G_{X}({X}_i^{1:N})G_{U}({U}_i^{1:N})G_{V}({V}_i^{1:N})
   },\\\nonumber
    &\phantom{----}  q_{G_{X}({X}^{1:N})G_{U}({U}^{1:N})G_{V}({V}^{1:N})
   })\\\nonumber
    &\phantom{--}+ \mathbb{V}(q_{G_{X}({X}^{1:N})G_{U}({U}^{1:N})G_{V}({V}_i^{1:N})
    },  p^{unif}_{\bar E}p^{unif}_{\bar D}p^{unif}_{\bar F})
    \\\nonumber
     & \stackrel{(c)}\leq \mathbb{V}(\widetilde p_{{X}_i^{1:N}{U}_i^{1:N}{V}_i^{1:N}
     },  q_{{X}^{1:N}{U}^{1:N}{V}^{1:N}
     })+ \textcolor{black}{\delta^{(0)}(N)}
     \\ 
          & \stackrel{(d)}\leq  \delta_i(N)+ \textcolor{black}{\delta^{(0)}(N)} \label{lem4eqn9},
   \end{align}
where $(a)$ holds by the triangle inequality, \textcolor{black}{$(b)$ holds because $\widetilde E_{i+1}\triangleq G_X(\widetilde X_i^{1:N}), \widetilde D_{i+1}\triangleq G_U(\widetilde U_i^{1:N}), \widetilde F_{i+1}\triangleq G_V(\widetilde V_i^{1:N})$ by Line $5$ of Algorithm \ref{alg:encoding_11} and Algorithm \ref{alg:encoding_21}
,} $(c)$ holds by the data processing inequality and Lemma \ref{lem3}, $(d)$ holds by the induction hypothesis.
By combining (\ref{lem4eqn7}) and \eqref{lem4eqn9}, we~have
\begin{align}\nonumber
   & \mathbb{V}(\widetilde p_{U_{i+1}^{1:N}V_{i+1}^{1:N}X_{i+1}^{1:N}Y_{i+1}^{1:N}Z_{i+1}^{1:N}}, q_{U^{1:N}V^{1:N}X^{1:N}Y^{1:N}Z^{1:N}}) \\\nonumber & \leq 3 ( \delta(N)+\delta_i(N)+ \textcolor{black}{\delta^{(0)}(N)})\\
  \nonumber &=\textcolor{black}{\delta_{i+1}(N)}. \hfill 
     \end{align}\end{proof}
  The next lemma shows that the recycled randomness in Block $i\in \llbracket 2, k \rrbracket$ is almost independent of the channel output in Block $i-1$.
      \begin{lem}\label{lem5}
           For $i\in \llbracket 2, k \rrbracket$, we have
           \begin{align*}
           \mathbb{V}(\widetilde{p}_{{Z}_{i-1}^{1:N} E_i D_i F_i}, \widetilde{p}_{{Z}_{i-1}^{1:N}}\widetilde{p}_{ E_i  D_i F_i}) \leq \delta^{(1)}_i(N),
 \end{align*} 
where $\delta^{(1)}_i(N) \triangleq 4 \delta_{i-1}(N) + 2 \textcolor{black}{\delta^{(0)}(N)}$.
\end{lem}
 \begin{proof}
We have 
\begin{align}\nonumber
&\mathbb{V}(\widetilde{p}_{{Z}_{i-1}^{1:N} E_i D_i F_i}, \widetilde{p}_{{Z}_{i-1}^{1:N}}\widetilde{p}_{ E_i  D_i F_i})\\\nonumber
& \stackrel{(a)}\leq \mathbb{V}(\widetilde p_{Z_{i-1}^{1:N}E_{i}D_{i}F_{i}},  \widetilde p_{Z_{i-1}^{1:N}} p^{unif}_{\bar E}p^{unif}_{\bar D}p^{unif}_{\bar F})\\\nonumber
    &\phantom{--}+\mathbb{V}(\widetilde p_{Z_{i-1}^{1:N}} p^{unif}_{\bar E}p^{unif}_{\bar D}p^{unif}_{\bar F}, \widetilde{p}_{{Z}_{i-1}^{1:N}}\widetilde{p}_{ E_i  D_i F_i})\\\nonumber
& \leq 2 \mathbb{V}(\widetilde p_{Z_{i-1}^{1:N}E_{i}D_{i}F_{i}},  \widetilde p_{Z_{i-1}^{1:N}} p^{unif}_{\bar E}p^{unif}_{\bar D}p^{unif}_{\bar F})\\\nonumber
 & \stackrel{(b)}\leq 2\left(\mathbb{V}(\widetilde p_{E_{i}D_{i}F_{i}Z^{1:N}_{i-1}},  q_{EDFZ^{1:N}})\right.\\\nonumber
    & \phantom{--}\left.+\mathbb{V}(q_{EDFZ^{1:N}}, p^{unif}_{\bar E} p^{unif}_{\bar D} p^{unif}_{\bar F} q_{Z^{1:N}}\right.
    )\\ \nonumber& \phantom{--}\left.+\mathbb{V}(p^{unif}_{\bar E} p^{unif}_{\bar D} p^{unif}_{\bar F} q_{Z^{1:N}}
    ,  p^{unif}_{\bar E} p^{unif}_{\bar D} p^{unif}_{\bar F} \widetilde p_{Z_{i-1}^{1:N}}
    )\right)\\\nonumber
   & \stackrel{(c)}\leq 2 \left(\mathbb{V}(\widetilde p_{{X}_{i-1}^{1:N}{U}_{i-1}^{1:N}{V}_{i-1}^{1:N}Z_{i-1}^{1:N}},  q_{{X}^{1:N}{U}^{1:N}{V}^{1:N}Z^{1:N}})\right.\\\nonumber
    &\phantom{--}\left.+ \mathbb{V}(q_{EDFZ^{1:N}
    },  p^{unif}_{\bar E}p^{unif}_{\bar D}p^{unif}_{\bar F}q_{Z^{1:N}}
    )\right.\\\nonumber
    &\phantom{--}\left.+ \mathbb{V}(q_{Z^{1:N}}, \widetilde p_{Z_{i-1}^{1:N}})\right)
    \\
\nonumber& \stackrel{(d)}\leq 2 (2\mathbb{V}(\widetilde p_{{X}_{i-1}^{1:N}{U}_{i-1}^{1:N}{V}_{i-1}^{1:N}Z_{i-1}^{1:N}},  q_{{X}^{1:N}{U}^{1:N}{V}^{1:N}Z^{1:N}})\! +\! \textcolor{black}{\delta^{(0)}(N)}) \\
& \stackrel{(e)}\leq 4 \delta_{i-1}(N) + 2 \textcolor{black}{\delta^{(0)}(N)},\nonumber
 \end{align}
 where $(a)$ and $(b)$ hold by the triangle inequality, $(c)$ holds by the data processing inequality using \eqref{eqnedf} and $\widetilde E_{i}\triangleq G_X(\widetilde X_{i-1}^{1:N}), \widetilde D_{i}\triangleq G_U(\widetilde U_{i-1}^{1:N}), \widetilde F_{i}\triangleq G_V(\widetilde V_{i-1}^{1:N})$ from Line $5$ of Algorithm \ref{alg:encoding_11} and Algorithm \ref{alg:encoding_21},  $(d)$~holds by \eqref{eqnedf} and Lemma~\ref{lem3}, $(e)$~holds by Lemma~\ref{lemblock}.
  \end{proof}  
    The next lemma shows that the recycled randomness in Block $i\in \llbracket 2, k \rrbracket$ is almost independent of the channel outputs in Blocks $1$ to $i-1$ considered jointly.
\begin{lem} \label{lemnew}
   For $i\in \llbracket 2, k \rrbracket$, we have$$
\mathbb{V} \left( \widetilde{p}_{Z_{1:i-1}^{1:N}D_i E_iF_i},   \widetilde{p}_{Z_{1:i-1}^{1:N}} \widetilde{p}_{ D_i E_iF_i } \right) \leq \delta_i^{(2)}(N),
$$
where $\delta_i^{(2)}(N)\triangleq (2^{i-1}-1) (4 \delta_{i-1}(N)+ 2 \textcolor{black}{\delta^{(0)}(N)})$.
\end{lem}
\begin{proof}
We prove the result by induction. The lemma is true for $i=2$ by Lemma \ref{lem5}. Assume now that the lemma holds for $i \in \llbracket 2,k-1 \rrbracket$. Then, for $i \in \llbracket 3,k \rrbracket$, we have
\begin{align*}
   \mathbb{V} \left( \widetilde{p}_{Z_{1:i-2}^{1:N}D_{i-1} E_{i-1}F_{i-1}},   \widetilde{p}_{Z_{1:i-2}^{1:N}} \widetilde{p}_{ D_{i-1} E_{i-1} F_{i-1} } \right) \leq \delta_{i-1}^{(2)}(N). 
\end{align*}
We have 
   \begin{align*}
    &  \mathbb{V} \left( \widetilde{p}_{Z_{1:i-1}^{1:N}D_i E_iF_i},   \widetilde{p}_{Z_{1:i-1}^{1:N}} \widetilde{p}_{ D_i E_iF_i } \right)\\
    & \stackrel{(a)}{\leq}  
       \mathbb{V} \left( \widetilde{p}_{Z_{1:i-1}^{1:N}D_{i} E_{i}F_{i}} , \widetilde{p}_{Z_{1:i-2}^{1:N}} \widetilde{p}_{ Z_{i-1}^{1:N} D_{i} E_{i}F_{i} } \right) \\& \phantom{--}+ \mathbb{V} \left( \widetilde{p}_{Z_{1:i-2}^{1:N}} \widetilde{p}_{ Z_{i-1}^{1:N} D_{i} E_{i}F_{i} }, \widetilde{p}_{Z_{1:i-2}^{1:N}} \widetilde{p}_{Z_{i-1}^{1:N}} \widetilde{p}_{ D_i E_iF_i } \right) \\ & \phantom{--} + \mathbb{V} \left( \widetilde{p}_{Z_{1:i-2}^{1:N}} \widetilde{p}_{Z_{i-1}^{1:N}} \widetilde{p}_{ D_i E_iF_i }, \widetilde{p}_{Z_{1:i-1}^{1:N}} \widetilde{p}_{ D_i E_iF_i } \right)\\
            & =  \mathbb{V} \left( \widetilde{p}_{Z_{1:i-1}^{1:N}D_{i} E_{i}F_{i}} , \widetilde{p}_{Z_{1:i-2}^{1:N}} \widetilde{p}_{ Z_{i-1}^{1:N} D_{i} E_{i}F_{i} } \right)
        \\ & \phantom{--} +\mathbb{V} \left(  \widetilde{p}_{Z_{i-1}^{1:N}D_{i} E_{i}F_{i}}  , \widetilde{p}_{Z_{i-1}^{1:N}} \widetilde{p}_{ D_i E_iF_i }\right)\\ & \phantom{--}+ \mathbb{V} \left( \widetilde{p}_{Z_{1:i-2}^{1:N}} \widetilde{p}_{Z_{i-1}^{1:N}} , \widetilde{p}_{Z_{1:i-1}^{1:N}}  \right)\\
    & \stackrel{(b)}{\leq}    \mathbb{V} \left( \widetilde{p}_{Z_{1:i-1}^{1:N}D_{i} E_{i}F_{i}} , \widetilde{p}_{Z_{1:i-2}^{1:N}} \widetilde{p}_{ Z_{i-1}^{1:N} D_{i} E_{i}F_{i} } \right)
        \\ & \phantom{--} + \mathbb{V} \left( \widetilde{p}_{Z_{1:i-2}^{1:N}} \widetilde{p}_{Z_{i-1}^{1:N}} , \widetilde{p}_{Z_{1:i-1}^{1:N}}  \right) + \delta^{(1)}_i(N) \\
         & \stackrel{(c)} \leq    2\mathbb{V} \left( \widetilde{p}_{Z_{1:i-1}^{1:N}D_{i-1:i} E_{i-1:i}F_{i-1:i}} ,\right.\\ & \phantom{----}\left. \widetilde{p}_{Z_{1:i-2}^{1:N}} \widetilde{p}_{ Z_{i-1}^{1:N} D_{i-1:i} E_{i-1:i}F_{i-1:i} } \right) + \delta^{(1)}_i(N) 
          \\
            & \stackrel{(d)}{=}  2 \mathbb{V} \left( \widetilde{p}_{Z_{1:i-2}^{1:N}D_{i-1} E_{i-1}F_{i-1}} \widetilde{p}_{Z_{i-1}^{1:N}D_{i} E_{i}F_{i} |D_{i-1} E_{i-1}F_{i-1}}  ,\right.\\ & \phantom{----}\left. \widetilde{p}_{Z_{1:i-2}^{1:N}} \widetilde{p}_{ Z_{i-1}^{1:N} D_{i-1:i} E_{i-1:i}F_{i-1:i} } \right) + \delta^{(1)}_i(N)
         \\
            & =  2 \mathbb{V} \left( \widetilde{p}_{Z_{1:i-2}^{1:N}D_{i-1} E_{i-1}F_{i-1}}   , \widetilde{p}_{Z_{1:i-2}^{1:N}} \widetilde{p}_{  D_{i-1} E_{i-1}F_{i-1} } \right) \\ & \phantom{--}+ \delta^{(1)}_i(N)\\
                          & \stackrel{(e)}{\leq}  \delta^{(1)}_i(N)  + 2 \delta_{i-1}^{(2)}(N)\\
                          & \leq \delta_{i}^{(2)}(N),
  \end{align*}
 where $(a)$ holds by the triangle inequality, $(b)$  holds  by Lemma~\ref{lem5}, \textcolor{black}{$(c)$ follows from the data processing inequality}, $(d)$ holds by the  Markov chain $(\widetilde{D}_{i},\widetilde{E}_{i},\widetilde{F}_{i}, \widetilde{Z}_{i-1}^{1:N}) - (\widetilde{D}_{i-1},\widetilde{E}_{i-1},\widetilde{F}_{i-1}) - \widetilde{Z}_{1:i-2}^{1:N}$, $(e)$ holds by the induction hypothesis. 
\end{proof}
     The next lemma shows that the channel outputs of all the blocks are asymptotically independent.
    \begin{lem}
    We have
   $$
\mathbb{V} \left( \widetilde{p}_{Z_{1:k}^{1:N}},\prod_{i =1}^k  \widetilde{p}_{Z_{i}^{1:N}} \right)\leq (k-1)  \delta_k^{(2)} (N),
$$ where $\delta_k^{(2)} (N)$ is defined in Lemma \ref{lemnew}. \label{lem8}
  \end{lem}
 \begin{proof}
  We have
 \begin{align*}
    & \mathbb{V} \left( \widetilde{p}_{Z_{1:k}^{1:N}},\prod_{i =1}^k  \widetilde{p}_{Z_{i}^{1:N}} \right)\\
     & \stackrel{(a)}{\leq} \sum_{i=2}^k \mathbb{V} \left( \widetilde{p}_{Z_{1:i}^{1:N}} \prod_{j=i+1}^k \widetilde{p}_{Z_{j}^{1:N}},   \widetilde{p}_{Z_{1:i-1}^{1:N}} \prod_{j=i}^k \widetilde{p}_{Z_{j}^{1:N}} \right)\\
          & = \sum_{i=2}^k \mathbb{V} \left( \widetilde{p}_{Z_{1:i}^{1:N}} ,   \widetilde{p}_{Z_{1:i-1}^{1:N}}  \widetilde{p}_{Z_{i}^{1:N}} \right)\\
    & \leq \sum_{i=2}^k \mathbb{V} \left( \widetilde{p}_{Z_{1:i}^{1:N}D_i E_iF_i},  \widetilde{p}_{Z_{i}^{1:N} D_i E_iF_i } \widetilde{p}_{Z_{1:i-1}^{1:N}} \right)\\
    & \stackrel{(b)}{=} \sum_{i=2}^k \mathbb{V} \left( \widetilde{p}_{Z_{1:i-1}^{1:N}|D_i E_iF_i}\widetilde{p}_{Z_{i}^{1:N}D_i E_iF_i}, \widetilde{p}_{Z_{i}^{1:N} D_i E_iF_i } \widetilde{p}_{Z_{1:i-1}^{1:N}} \right)\\
    & = \sum_{i=2}^k \mathbb{V} \left( \widetilde{p}_{Z_{1:i-1}^{1:N}D_i E_iF_i},   \widetilde{p}_{Z_{1:i-1}^{1:N}} \widetilde{p}_{ D_i E_iF_i } \right)\\
    & \stackrel{(c)}{\leq} \sum_{i=2}^k \delta_i^{(2)} (N) \\
        & \leq (k-1) \max_{j \in \llbracket 2,k \rrbracket }\delta_j^{(2)} (N), 
  \end{align*}
 where $(a)$ holds by the triangle inequality, $(b)$ holds by the Markov chain $\widetilde{Z}_i^{1:N}   - (\widetilde{D}_i,\widetilde{E}_i,\widetilde{F}_i)- \widetilde{Z}_{1:i-1}^{1:N} $, $(c)$ holds by Lemma \ref{lemnew}.
\end{proof}
                        We  now  show  that  the  target  output  distribution  is  well approximated jointly over all blocks.
             \begin{lem} \label{lemindep}  For Block $i\in \llbracket 1, k \rrbracket$,
         we have
$$\mathbb{V} \left( \widetilde{p}_{Z_{1:k}^{1:N}},   q_{Z^{1:kN}}  \right) \leq    (k-1)  \delta_k^{(2)} (N)+ k \delta_{k}(N),
$$    where $\delta_k^{(2)} (N)$ is defined in Lemma \ref{lemnew} and $\delta_k(N)$ is defined in Lemma \ref{lemblock}.  \end{lem}\label{lem9}
         \begin{proof} We have
      \begin{align}\nonumber
  & \mathbb{V} ( \widetilde{p}_{Z_{1:k}^{1:N}},   q_{Z^{1:kN}} )  \\\nonumber
  & \stackrel{(a)}\leq (k-1)   \delta_k^{(2)} (N)+ \mathbb{V} (  \prod_{i =1}^k  \widetilde{p}_{Z_{i}^{1:N}}, q_{Z^{1:kN}} )\\\nonumber
   & \stackrel{(b)}\leq  (k-1) \delta_k^{(2)} (N) + \mathbb{V} (\widetilde{p}_{Z_{1}^{1:N}} \prod_{i =2}^k  \widetilde{p}_{Z_{i}^{1:N}}, q_{Z^{1:N}} \prod_{i =2}^k  \widetilde{p}_{Z_{i}^{1:N}} )\\\nonumber & \phantom{--}+ \mathbb{V} (  q_{Z^{1:N}} \prod_{i =2}^k  \widetilde{p}_{Z_{i}^{1:N}}, q_{Z^{1:kN}} )\\\nonumber
   & \stackrel{(c)}\leq  (k-1)\delta_k^{(2)} (N)+ \delta_{1}(N)+  \mathbb{V} ( \prod_{i =2}^k  \widetilde{p}_{Z_{i}^{1:N}}, q_{Z^{1:(k-1)N}} )\\ \nonumber 
   & \stackrel{(d)}\leq (k-1)\delta_k^{(2)} (N)+  \sum_{i =1}^k \delta_{i}(N)\\ \nonumber
 &  \leq    (k-1) \delta_k^{(2)} (N)+ k \max_{j \in \llbracket 1,k \rrbracket} \delta_{j}(N),
\end{align}
where $(a)$ holds by the triangle inequality and Lemma \ref{lem8}, $(b)$ holds by the triangle inequality, $(c)$ holds by Lemma \ref{lemblock}, $(d)$ holds by induction.
\end{proof}
Finally, the next lemma shows that the encoding scheme of Section \ref{sec:sub1} achieves the desired rate-tuple.
\begin{lem}
  Let $\epsilon_0>0$.  For $k$ large enough and $\xi>0$, we~have
    \begin{align*}
      \displaystyle \lim_{{{N \rightarrow +\infty}}}R_1 &= I(X;ZU)+\epsilon_0+2\xi ,\\
      \displaystyle  \lim_{{{N \rightarrow +\infty}}}R_U &= I(U;Z)+\epsilon_0+2\xi ,\\
      \displaystyle  \lim_{{{N \rightarrow +\infty}}}R_V &= I(V;ZUX)+\epsilon_0+2\xi.
    \end{align*}\label{lem10}
   \end{lem} 
   \begin{proof}
    Let $k$ be such that  $
       \frac{ 1}{k}\max (H(X) , H(U),H(V)) <\epsilon_0. $ Then, by the definition of~$\epsilon_1$, we have
      \begin{align*}
          R_1 &= \frac{ \sum_{i=1}^k| E_i|}{kN}\\ 
          &=  \frac{N(H(X)+\epsilon_1)  + (k-1)N(I(X;ZU)+\epsilon_1)}{kN}\\
   & \leq    \frac{H(X)}{k}+ I(X;ZU)+\epsilon_1\\
 &  \leq   \epsilon_0+ I(X;ZU)+\epsilon_1\\
   & \xrightarrow{N \to +\infty}I(X;ZU)+ \epsilon_0+2\xi,
      \end{align*} 
       \begin{align*}
          R_U &= \frac{\sum_{i=1}^k | D_i|}{kN}\\ &=  \frac{N(H(U)+\epsilon_1   + (k-1) N(I(U;Z)+\epsilon_1)}{kN}\\
  & \leq    \frac{H(U)}{k}+ I(U;Z)+\epsilon_1\\
   &\leq    \epsilon_0+I(U;Z)+\epsilon_1\\
   & \xrightarrow{N \to +\infty}I(U;Z)+ \epsilon_0+2\xi,
      \end{align*}
       \begin{align*}
          R_V &= \frac{\sum_{i=1}^k | F_i|}{kN}\\ &=  \frac{N(H(V)+\epsilon_1) + (k-1)N(I(V;ZUX)+\epsilon_1)}{kN}\\
  & \leq    \frac{H(V)}{k}+ I(V;ZUX)+\epsilon_1\\
   & \leq    \epsilon_0+ I(V;ZUX)+\epsilon_1\\
   & \xrightarrow{N \to +\infty}I(V;ZUX)+\epsilon_0+2\xi.
      \end{align*} 
      
   \end{proof}
    \subsection{Coding scheme analysis for Case 2}
    
 For Case $2$, $U=\emptyset$ and $V=Y$, so that by  Lemma $\ref{lem10}$, the achieved rate pair is such that
 \begin{align*}
       \displaystyle \lim_{{{N \rightarrow +\infty}}}R_1 &= I(X;Z)+\epsilon_0+2\xi,\\
                    \displaystyle \lim_{{{N \rightarrow +\infty}}}R_2 &= \lim_{{{N \rightarrow +\infty}}} (R_V +R_U) \\&= I(Y;ZX)+\epsilon_0+2\xi\\
         &\stackrel{(a)}=I(Y; Z|X)+\epsilon_0+2\xi\\
         &\stackrel{(b)}=I(Y;Z)+\epsilon_0+2\xi,
 \end{align*}
where $(a)$ holds  by independence between $X$ and $Y$, and $(b)$ holds  because  $I(XY;Z) = I(X;Z)+ I(Y;Z)$ in Case $2$.  
     \section{Extension to more than two transmitters}\label{sec:extension}
   Consider a discrete memoryless multiple access  channel  $(\mathcal X_{\mathcal{L}},  q_{Z|X_{\mathcal{L}}}, \mathcal Z)$, where $\mathcal X_l=\{ 0,1\}$, $l\in\mathcal{L}\triangleq \llbracket 1,L\rrbracket$,  $\mathcal{Z}$ is a finite alphabet, and $X_{\mathcal{L}} \triangleq (X_l)_{l \in \mathcal{L}}$. The definitions in Section~\ref{sec:ps} immediately extend to this multiple access channel with $L$ transmitters and we have the following counterpart of Theorem \ref{thm1}.
   \begin{thm} \label{thresol}
 We have $\mathcal R_{q_Z} = \mathcal{R}'_{q_Z}$ with
 \begin{align*}
   \mathcal R'_{q_Z} \triangleq {\bigcup_{p_T,(q_{X_l|T})_{l\in\mathcal{L}}} }\!\!\!\!\!\!\!\!\!\{(R_l)_{l\in\mathcal{L}}:  I(X_{\mathcal{S}};Z|T)&\leq R_{\mathcal{S}}, \forall \mathcal{S} \subseteq \mathcal{L} \},
 \end{align*}
    where $p_T$ is defined  over  $\mathcal{T}\triangleq \llbracket1,|\mathcal{Z}|+2^L-1\rrbracket$ and  $(q_{X_l|T})_{l\in \mathcal{L}}$ are such that, for any $t\in \mathcal{T}$ and  $z\in \mathcal{Z}$,
    \begin{align*}
        q_Z (z)= \sum_{{x_{\mathcal{L}} \in \mathcal{X}_{\mathcal{L}}}}   q_{Z|X_{\mathcal{L}}}(z|x_{\mathcal{L}}) \prod_{l\in\mathcal{L}}q_{X_l|T}(x_l|t).
    \end{align*} 
        \end{thm}
   \begin{proof}
   The converse is an immediate extension of the converse of Theorem \ref{thm1} from \cite{frey2017mac}. The achievability follows from Theorem \ref{thm5}.
   \end{proof}
   
   \begin{thm} \label{thm5}
The coding scheme presented in Section \ref{sec:mult}, which solely relies on source resolvability codes, used as black boxes, and two-universal hash functions, achieves the entire  multiple access channel resolvability region $\mathcal R_{q_Z}$ of Theorem~\ref{thresol} for any discrete memoryless multiple access channel with binary input alphabets. 
\end{thm}

\subsection{Achievability Scheme} \label{sec:mult}
In the following, we use the notation $X_{\mathcal{S}} \triangleq (X_l)_{l\in \mathcal{S}}$ for $\mathcal{S} \subseteq \mathcal{L}$, and $X_{1:l} \triangleq X_{\llbracket 1,l \rrbracket}$ for $l \in \mathcal{L}$.
Let $p_{X_{\mathcal{L}}} \triangleq \prod_{l\in\mathcal{L}} p_{X_l}$.  We will show the achievability of the region  
 \begin{align*}
\mathcal{R} \left(p_{X_{\mathcal{L}}}\right) \triangleq \{(R_l)_{l\in\mathcal{L}}:  I(X_{\mathcal{S}};Z)& \leq R_{\mathcal{S}}, \forall \mathcal{S} \subseteq \mathcal{L}
    \},
\end{align*}
which reduces to showing the achievability of the rate-tuple $(I(X_l; Z|X_{1:l-1}))_{l\in\mathcal{L}}$. Indeed, the set function $\mathcal{S} \mapsto - I(X_{\mathcal{S}};Z)$ is submodular, e.g., \cite{chou2018polar}, and the region $\mathcal{R} \left(p_{X_{\mathcal{L}}}\right) $ thus forms a contrapolymatroid \cite{edmonds2003submodular} whose dominant face is the convex hull of its extreme points given by $\{ (I(X_{\sigma(l)}; Z|X_{\{ \sigma(i) : i \in \llbracket 1 , l-1 \rrbracket \}}))_{l\in\mathcal{L}} : \sigma \in \frak{S}(L)  \}$, where $\frak{S}(L)$ is the symmetric group over $\mathcal{L}$. By time-sharing and symmetry of the extreme points,  the achievability of the dominant face reduces to showing the achievability of one extreme point, which without loss of generality can be chosen as $(I(X_l; Z|X_{1:l-1}))_{l\in\mathcal{L}}$.

The encoding scheme to achieve $(I(X_l; Z|X_{1:l-1}))_{l\in\mathcal{L}}$ operates over $k \in \mathbb{N}$ blocks of length $N$. In this section, we use the double subscripts notation ${X_{l,i}}$, where the first subscript corresponds to Transmitter~$l \in \mathcal{L}$ and the second subscript corresponds to Block $i \in \llbracket 1,k \rrbracket$.  
The encoding at Transmitter~$l\in \mathcal{L}$ is described in Algorithm~$\ref{algext}$ and uses
    \begin{itemize}
    \item A hash function $G_{X_{l}}:\{0,1\}^{N} \longrightarrow \{0,1\}^{r_{{X_{l}}}}$ chosen uniformly at random in a family of two-universal hash functions, where the output length of the hash function $G_{X_{l}}$ is defined as follows  
   \begin{align}
    r_{X_{l}} \triangleq N(H(X_{l}|ZX_{1:l-1})- {\epsilon_2}/{2}).\label{c2def1}
\end{align}
\item A source resolvability code  for the discrete memoryless source $(\mathcal{X}_{l},q_{X_{l}})$ with encoder function $e_N^{X_{{l}}}$ and rate $H(X_{l})+\frac{\epsilon_2 }{2}$, where $\epsilon_2 \triangleq 2( \delta^*_{\mathcal{L}}(N)+\xi)$, $\delta^*_{\mathcal{L}}(N)\triangleq \log (|\mathcal{X}_{\mathcal{L}}|+3)\sqrt{\frac{2}{N}(L+\log N)}$, $\xi>0$, such that the distribution of the encoder output $\widetilde{p}_{X_{l}^{1:N}}$ satisfies $\mathbb{V}(\widetilde{p}_{X_{l}^{1:N}},q_{X_{l}^{1:N}})\leq \delta(N)$, where $\delta(N)$ is such that $\lim_{N \to +\infty} \delta(N) =0$.
\end{itemize}
In Algorithm~$\ref{algext}$ and for any $l \in \mathcal{L}$, the hash function output ${\widetilde E_{{l}, i}}$, $i\in \llbracket 2,k\rrbracket$, with length $r_{X_{l}}$ corresponds to recycled randomness from Block $i-1$. 
 \begin{algorithm}[h]
  \caption{Encoding algorithm at Transmitter $l\in \mathcal{L}$}
   \label{algext}
    \begin{algorithmic}   [1] 
     \REQUIRE A vector $ E_{l, 1}$ of $N (H(X_l)+\epsilon_2)$ uniformly distributed bits,   and for $i \in \llbracket 2,k \rrbracket$, a vector $E_{l, {i}}$ of $N (I(X_l; ZX_{1:l-1})+\epsilon_{2})$ uniformly distributed bits. 
                                 \FOR{Block $i=1$ to $k$}
                         \IF{$i=1$}
                         \STATE Define ${\widetilde X_{l,1}}^{1:N}\triangleq e^{X_l}_N(E_{l, 1})$ 
                         \ELSIF{$i>1$}
      \STATE  Define $\widetilde E_{l, {i}} \triangleq G_{X_l} (\widetilde X_{l, {i-1}}^{1:N})$
       \STATE Define $\widetilde X_{l, i}^{1:N}\triangleq e^{X_l}_N(\widetilde E_{l, {i}} \lVert {E_{l,i}})$
                                  \ENDIF
                                  \STATE Send ${\widetilde X_{l, i}}^{1:N}$  over the channel
                        \ENDFOR 
       \end{algorithmic}
\end{algorithm}
\subsection{Achievability Scheme Analysis}
 For convenience, define, for any $l \in \mathcal{L}$, $\widetilde E_{l,1} \triangleq \emptyset$. Let 
$ \widetilde p_{ E_{1:L, i}  X^{1:N}_{1:L, i} Z^{1:N}_{i}}$
       denote the joint probability distribution of the random variables $\widetilde E_{l, i}, \widetilde X^{1:N}_{l, i}$, and $\widetilde Z^{1:N}_{i}$, $l \in \mathcal{L}$, created in Block $i \in \llbracket 1,k\rrbracket$ of the coding scheme of Section {\ref{sec:mult}}.

 We prove in the following lemma that in Block $i\in \llbracket 2,k \rrbracket$, if the inputs $\widetilde X_{1:L, i-1}^{1:N}$ of the hash functions $(G_{X_l})_{l\in \mathcal{L}}$ are replaced by $X_{1:L}^{1:N}$ distributed according to $q_{X^{1:N}_{1:L}} \triangleq \prod_{i=1}^N q_{X_{1:L}}$, then the outputs of these hash functions are almost jointly uniformly distributed. Define
 \begin{align*}
     G_{X_{1:L}}(X^{1:N}_{1:L})\triangleq  \left(G_{X_{l}}(X^{1:N}_{l})\right)_{l \in \mathcal{L}}.
 \end{align*}
 \begin{lem} \label{c2lem3} 

 Let $ p^{unif}_{\bar E_{1:L}}$  denote the uniform distribution over $\{0,1\}^{ \sum_{l\in \mathcal{L}}r_{X_{l}}}$.  
  Then, we have
       \begin{align*}
           & \mathbb{V} \left(q_{G_{X_{1:L}}(X_{1:L}^{1:N})Z^{1:N}},  p^{unif}_{\bar E_{1:L}}q_{Z^{1:N}} \right)  \leq     \delta^{*(0)}(N),
     \end{align*}
                where $\delta^{*(0)}(N) \triangleq 2/N +  {2^{\frac{L}{2}} 2^{-\frac{N\xi}{2}}} $.       
                   \end{lem}                \begin{proof} 
                                      Using Lemma \ref{lems1} in Appendix~\ref{appA}, with the substitutions $\mathcal{A}\leftarrow \mathcal{L}, T_{\mathcal{A}}^{1:N} \leftarrow X_{\mathcal{L}}^{1:N}$, applied to the product distribution $q_{X_{\mathcal{L}}^{1:N}Z^{1:N}}$, there exists a subnormalized non-negative function $w_{{X_{\mathcal{L}}^{1:N}}Z^{1:N}}$ such that for any $\mathcal{S}\subseteq \mathcal{L}$  
 \begin{align}
           \mathbb{V}({w}_{ X_{1:L}^{1:N}Z^{1:N}}, q_{ X_{1:L}^{1:N} Z^{1:N}}) &\leq   1/N ,\label{c2lem3eqn3}\\
                     H_{\infty}({w}_{{{X}_{\mathcal{S}}^{1:N}} Z^{1:N}}|q_{ Z^{1:N}}) &\geq N H({{ X}_{\mathcal{S}}}| Z)-N \delta^{*}_{\mathcal{S}}(N) ,\label{c2lem3eqn4}
       \end{align} 
       where the min-entropy $H_{\infty}({w}_{{{X}_{\mathcal{S}}^{1:N}} Z^{1:N}}|q_{ Z^{1:N}})$ is defined in Lemma \ref{lems1} in Appendix~\ref{appA}, and  $\delta^*_{\mathcal{S}}(N)\triangleq \log (\lvert\mathcal{X}_{\mathcal{S}} \rvert+3) \sqrt{\frac{2}{N}(L+\log N)}$.
           Let $q_{E_{1:L}}$ define the  distribution of 
        \begin{align}
            E_{1:L} \triangleq G_{X_{1:L}} ( X_{1:L}^{1:N}),\label{eqndefel}
        \end{align}
         where  $X_{1:L}^{1:N}$ is distributed according to $q_{X_{1:L}^{1:N}}$. We have 
               \begin{align*}\nonumber
      &  \mathbb{V}(q_{E_{1:L}Z^{1:N}},  p^{unif}_{\bar E_{1:L}}q_{Z^{1:N}})  \\ \nonumber       &\stackrel{(a)}\leq \mathbb{V}(q_{E_{1:L}Z^{1:N}}, w_{E_{1:L} Z^{1:N}})+   \mathbb{V}(w_{E_{1:L}Z^{1:N}},  p^{unif}_{\bar E_{1:L}}q_{Z^{1:N}})   \\ \nonumber
             & \stackrel{(b)} \leq \mathbb{V}( q_{ X_{1:L}^{1:N}  Z^{1:N}},{w}_{ X_{1:L}^{1:N}Z^{1:N} })+\mathbb{V}(w_{E_{1:L}Z^{1:N}},  p^{unif}_{\bar E_{1:L}}q_{Z^{1:N}}) \\ \nonumber
                 & \stackrel{(c)} \leq 1/N +   \mathbb{V}(w_{E_{1:L} Z^{1:N}},  p^{unif}_{\bar E_{1:L}} w_{Z^{1:N}})\\ & \phantom{--}+ \mathbb{V}( p^{unif}_{\bar E_{1:L}} w_{Z^{1:N}},  p^{unif}_{\bar E_{1:L}} q_{Z^{1:N}}) \\ \nonumber
           & \stackrel{(d)} \leq 2/N +   \mathbb{V}(w_{E_{1:L}Z^{1:N}},  p^{unif}_{\bar E_{1:L}} w_{Z^{1:N}}) \\ \nonumber
      & \stackrel{(e)}  \leq 2/N +  \sqrt{{ \displaystyle\sum_{\substack{{\mathcal S\subseteq\mathcal L},  {\mathcal S \neq \emptyset}}}}2^{r_{X_{\mathcal S}}-H_{\infty}({w}_{{{X}_{\mathcal{S}}^{1:N}} Z^{1:N}}|q_{ Z^{1:N}})}}\\
          \nonumber
            &\stackrel{(f)}\leq 2/N + {{{\sqrt{{ \displaystyle\sum_{\substack{{\mathcal S\subseteq\mathcal L}, {\mathcal S \neq \emptyset}}}}2^{r_{X_{\mathcal S}}-NH( X_{\mathcal S}|Z)+N \delta^{*}_{\mathcal{L}}(N)}}}}} \\ 
     &\stackrel{(g)}=2/N +  \left( \displaystyle\sum_{\substack{{\mathcal S\subseteq\mathcal L}, {\mathcal S \neq \emptyset}}} 2^{\sum_{l \in \mathcal S}(N(H(X_l|ZX_{1:l-1})-\frac{\epsilon_2}{2}))}\right.    \\ \nonumber        & \phantom{--} \left.   \times   2^{-N \sum_{l \in \mathcal S} H( X_{l}|ZX_{\llbracket 1,l-1\rrbracket \cap \mathcal{S}})+N \delta^{*}_{\mathcal{L}}(N)} \right)^{1/2}\\
          &\stackrel{(h)} \leq 2/N + {{{\sqrt{{ \displaystyle\sum_{\substack{{\mathcal S\subseteq\mathcal L}, {\mathcal S \neq \emptyset}}}}2^{\sum _{l \in \mathcal{S}}N(-\frac{\epsilon_2}{2}+ \delta^{*}_{\mathcal{L}}(N))}}}}}\\
             & \stackrel {(i)} \leq  2/N + \sqrt{ {2^L2^{-N\xi}} }  \xrightarrow{N \to +\infty}0,
         \end{align*}
    where $(a)$ holds by the triangle inequality, $(b)$ holds by~\eqref{eqndefel} and the data processing inequality, $(c)$  holds by~(\ref{c2lem3eqn3}) and the triangle inequality,   $(d)$ holds by (\ref{c2lem3eqn3}), $(e)$ holds by Lemma~\ref{lem11} in Appendix~\ref{appA}, $(f)$ holds by  (\ref{c2lem3eqn4}) and because for any $\mathcal{S}\subseteq \mathcal{L}$, $\delta^{*}_{\mathcal{S}}(N) \leq \delta^{*}_{\mathcal{L}}(N)$, $(g)$ holds by (\ref{c2def1}) and the chain rule, $(h)$ holds because conditioning reduces entropy, $(i)$ holds because $|\mathcal{S}|\geq 1$ and $\epsilon_2= 2( \delta^{*}_{\mathcal{L}}(N)+\xi)$. 
    \end{proof}
    We now show that in each encoding block, the random variables induced by the coding scheme approximate well the target distribution.
          \begin{lem} \label{lemblock2}
           For 
           Block $i\in \llbracket 1, k \rrbracket$, 
 \begin{align}\label{lemblock2eqn4}
\mathbb{V}(\widetilde p_{X_{1:L,i}^{1:N}Z_i^{1:N}}, q_{X_{1:L}^{1:N}Z^{1:N}})\leq \delta^{*}_i(N),
 \end{align}
where $\delta^{*}_i(N) \triangleq L(\delta(N) + \delta^{*(0)}(N))(\frac{L^{i}-1}{L-1}) + L^{i+1} \delta(N)$.
 \end{lem}
  \begin{proof} 
  We prove the result by induction. For $i=1$, 
we have
\begin{align}
& \mathbb{V}(\widetilde p_{X_{1:L, 1}^{1:N}Z_1^{1:N}}, q_{X_{1:L}^{1:N}Z^{1:N}}) \nonumber \\\nonumber
&=\mathbb {V}( \widetilde p_{Z_1^{1:N}|X_{1:L, 1}^{1:N}}\widetilde p_{X_{1:L, 1}^{1:N}}, q_{Z^{1:N}|X_{1:L}^{1:N}} q_{X_{1:L}^{1:N}})\\\nonumber
&\stackrel{(a)}=\mathbb {V}( \widetilde p_{X_{1:L, 1}^{1:N}},  q_{X_{1:L}^{1:N}})\\
&\stackrel{(b)}\leq \sum_{l \in \mathcal{L}}\mathbb {V}( \widetilde p_ {X_{l, 1}^{1:N}},  q_{X_{l}^{1:N}}) \label{mid2eqn} 
\\\nonumber
&\stackrel{(c)}\leq  L\delta(N),
\end{align}
where 
$(a)$ holds because $q_{Z^{1:N}\vert X_{1:L}^{1:N}} = \widetilde p_{Z_1^{1:N}\lvert X_{1:L, 1}^{1:N}}$, $(b)$ holds by the triangle inequality and because ${(\widetilde X_{l,1}^{1:N})}_{l \in \mathcal{L}}
$ are jointly independent, and $ ({X_{l}^{1:N}})_{l \in \mathcal{L}}
$ are jointly independent, $(c)$ holds by the source resolvability codes used at the transmitters because $\frac{|E_{l, 1}|}{N}>H(X_{l}) + \epsilon_2/2, l \in \mathcal{L}$.

Assume now that, for $i \in \llbracket 2,  k-1 \rrbracket$,
\eqref{lemblock2eqn4} holds. For any $l \in \mathcal{L}$ and $i\in \llbracket 2,k \rrbracket$, consider  $\bar E_{l, i}$ distributed according to $ {p^{unif}_{\bar E_{l}}}$, the uniform distribution over $\{0,1\}^{r_{X_{l}}}$, and let $p_{\bar X_{l, i}^{1:N}}$ denote the distribution of $\bar X_{l, i}^{1:N}\triangleq e^{X_{l}}_N(\bar E_{l, i}, E_{l, i})$.
For $i \in \llbracket 1,  k-1 \rrbracket$, we have
       \begin{align}\nonumber
&\mathbb{V}(\widetilde p_{X_{1:L, i+1}^{1:N}Z_{i+1}^{1:N}}, q_{X_{1:L}^{1:N}Z^{1:N}})\\\nonumber
&\stackrel{(a)} \leq \sum_{l \in \mathcal{L}} \mathbb{V}(\widetilde p_{X_{l, i+1}^{1:N}}, q_{X_{l}^{1:N}})
\\\nonumber
&\stackrel{(b)} \leq \sum_{l \in \mathcal{L}}\mathbb{V}(\widetilde p_{X_{l, i+1}^{1:N}}, p_{\bar X_{l, i+1}^{1:N}})
+\mathbb{V}( p_{\bar X_{l, i+1}^{1:N}}, q_{X_{l}^{1:N}})\\\nonumber
&\stackrel{(c)}\leq \sum_{l \in \mathcal{L}}\mathbb{V}(\widetilde p_{X_{l, i+1}^{1:N}}, p_{\bar X_{l,i+1}^{1:N}})+\delta(N) \\ \nonumber
&\stackrel{(d)}\leq \sum_{l \in \mathcal{L}}\mathbb{V}(\widetilde p_{E_{l, i+1}}, p^{unif}_{\bar E_{l}})+\delta(N) \\ \nonumber
        & \stackrel{(e)}\leq  \sum_{l \in \mathcal{L}}\left(\delta(N)+ \mathbb{V}(\widetilde p_{E_{l, i+1}},  q_{G_{X_{l}}({X^{1:N}_{l}})})\right. \\\nonumber & \phantom{--}\left.+\mathbb{V}(q_{G_{X_{l}}({X^{1:N}_{l}})},  p^{unif}_{\bar E_{l}}  )\right)  \\ \nonumber
        &\stackrel{(f)} =\sum_{l \in \mathcal{L}}\left(\delta(N)+ \mathbb{V}(\widetilde p_{G_{X_{l}}({X^{1:N}_{l, i}})},  q_{G_{X_{l}}({X^{1:N}_{l}})})\right.\\ \nonumber & \phantom{--}\left.+\mathbb{V}(q_{G_{X_{l}}({X^{1:N}_{l}})},  p^{unif}_{\bar E_{l}}  ) \right) \\ 
      \nonumber  & \stackrel{(g)}\leq  \sum_{l \in \mathcal{L}} \delta(N)+ \mathbb{V}(\widetilde p_{{X}^{1:N}_{l, i}},  q_{{X^{1:N}_{l}}})+ \delta^{*(0)}(N)
        \\ \nonumber
          & \stackrel{(h)}\leq \sum_{l \in \mathcal{L}}\delta(N)+\delta^{*}_i(N)+ \delta^{*(0)}(N) 
            \\ \nonumber
          & =L\left (\delta(N)+\delta^{*}_i(N)+ \delta^{*(0)}(N)\right) ,
\end{align}
where $(a)$ holds similar to (\ref{mid2eqn}), $(b)$ holds by the triangle inequality, $(c)$ holds by the source resolvability codes used at the transmitters because $\frac{|\bar E_{l, i}|+|E_{l, i}|}{N} = H(X_{l}) + \epsilon_2/2, l\in \mathcal{L}$, $(d)$ holds by the data processing inequality, $(e)$ holds by the triangle inequality, $(f)$ holds because for any $l\in \mathcal{L}$, $\widetilde E_{l, i+1}\triangleq G_{X_{l}}(X_{l, i}^{1:N})$ by Line $5$ of Algorithm \ref{algext}, $(g)$ holds by the data processing inequality and Lemma \ref{c2lem3}, $(h)$ holds by the induction hypothesis.
  \end{proof}

Next, we show that the recycled randomness in Block~$i \in \llbracket  2, k\rrbracket$ is almost independent from the channel outputs of Block~$i-1$.
  \begin{lem}
    For $i\in \llbracket 2, k \rrbracket$, we have
 \begin{align*}
\mathbb{V}(\widetilde{p}_{E_{1:L, i}{Z}_{i-1}^{1:N} }, \widetilde{p}_{ E_{1:L, i}} \widetilde{p}_{{Z}_{i-1}^{1:N}})\leq \delta^{*(1)}_i(N).
 \end{align*}where $\delta^{*(1)}_i(N) \triangleq 4 \delta_{i-1}^{*}(N) + 2 \delta^{*(0)}(N) $. 
 \label{c2lem6}
 \end{lem}
 \begin{proof}
We have 
\begin{align}\nonumber
&\mathbb{V}(\widetilde{p}_{E_{1:L,i}{Z}_{i-1}^{1:N}}, \widetilde{p}_{ E_{1:L,i}} \widetilde{p}_{{Z}_{i-1}^{1:N}})\\\nonumber
& \stackrel{(a)}\leq \mathbb{V}(\widetilde p_{E_{{1:L,i}} Z_{i-1}^{1:N}},  p^{unif}_{\bar E_{{1:L}}} \widetilde p_{Z_{i-1}^{1:N}} )\\\nonumber
& \phantom{--}+\mathbb{V}(p^{unif}_{\bar E_{1:L}}\widetilde p_{Z_{i-1}^{1:N}} , \widetilde{p}_{ E_{1:L,i} }\widetilde{p}_{{Z}_{i-1}^{1:N}})\\\nonumber
& \leq 2 \mathbb{V}(\widetilde p_{E_{{1:L,i}} Z_{i-1}^{1:N}},   p^{unif}_{\bar E_{1:L}}\widetilde p_{Z_{i-1}^{1:N}})\\\nonumber
 & \stackrel{(b)}\leq  2\left( \mathbb{V}(\widetilde p_{E_{1:L, i} Z_{i-1}^{1:N}},  q_{G_{X_{1:L}}({X^{1:N}_{1:L}})Z^{1:N}})\right.\\\nonumber & \phantom{--}\left. +\mathbb{V}(q_{G_{X_{1:L}}({X^{1:N}_{1:L})Z^{1:N}}},  p^{unif}_{\bar E_{1:L}}q_{Z^{1:N}})\right.\\\nonumber & \phantom{--}\left.+\mathbb{V}(p^{unif}_{\bar E_{1:L}}q_{Z^{1:N}},p^{unif}_{\bar E_{1:L}} \widetilde p_{Z^{1:N}_{i-1}})\right)  \\ \nonumber
   & \stackrel{(c)}\leq 2 (\mathbb{V}(\widetilde p_{{X}_{1:L, i-1}^{1:N}Z_{i-1}^{1:N}},  q_{{X}_{1:L}^{1:N}Z^{1:N}})+ {\delta^{*(0)}(N)}\\\nonumber & \phantom{--}+\mathbb{V}(q_{Z^{1:N}}, \widetilde p_{Z_{i-1}^{1:N}})) \\
& \stackrel{(d)}\leq 4 \delta^{*}_{i-1}(N) + 2 \textcolor{black}{\delta^{*(0)}(N)},
\nonumber
 \end{align}
 where $(a)$ and $(b)$ hold by the triangle inequality, $(c)$ holds by the data processing inequality because $\widetilde E_{1:L, i}\triangleq G_{X_{1:L}}(X_{1:L, i-1}^{1:N})$ by Line~$5$ of Algorithm~\ref{algext},
 and by Lemma~\ref{c2lem3}, $(d)$ holds by Lemma \ref{lemblock2}.
  \end{proof}  

Next, we show that the recycled randomness in Block $i\in \llbracket 2, k \rrbracket$ is almost independent of the channel outputs in Blocks~$1$ to $i-1$ considered jointly.
 \begin{lem} \label{c2lemnew}
   For $i\in \llbracket 2, k \rrbracket$, we have$$
\mathbb{V} \left( \widetilde{p}_{E_{1:L, i} Z_{1:i-1}^{1:N} }, \widetilde{p}_{E_{1:L, i}}  \widetilde{p}_{Z_{1:i-1}^{1:N}}  \right) \leq \delta_i^{*(2)}(N),
$$
where $\delta_i^{*(2)}(N) \triangleq (2^{i-1}-1) (4 \delta_{i-1}^{*}(N)+ 2 \delta^{*(0)}(N))$.
\end{lem}

\begin{proof}
We prove the result by induction. The lemma is true for $i=2$ by Lemma \ref{c2lem6}. Assume now that the lemma holds for $i \in \llbracket 2,k-1 \rrbracket$. Then, for $i \in \llbracket 3,k \rrbracket$, we have 
   \begin{align*}
    &  \mathbb{V} \left( \widetilde{p}_{Z_{1:i-1}^{1:N} E_{1:L, i}},   \widetilde{p}_{Z_{1:i-1}^{1:N}} \widetilde{p}_{E_{1:L, i}} \right)\\
    & \stackrel{(a)}{\leq}  
       \mathbb{V} \left( \widetilde{p}_{Z_{1:i-1}^{1:N} E_{1:L, i}} , \widetilde{p}_{Z_{1:i-2}^{1:N}} \widetilde{p}_{ Z_{i-1}^{1:N}  E_{1:L, i} } \right) .\\\nonumber & \phantom{--}+ \mathbb{V} \left( \widetilde{p}_{Z_{1:i-2}^{1:N}} \widetilde{p}_{ Z_{i-1}^{1:N} E_{1:L, i} }, \widetilde{p}_{Z_{1:i-2}^{1:N}} \widetilde{p}_{Z_{i-1}^{1:N}} \widetilde{p}_{ E_{1:L, i} } \right) \\ & \phantom{--} + \mathbb{V} \left( \widetilde{p}_{Z_{1:i-2}^{1:N}} \widetilde{p}_{Z_{i-1}^{1:N}} \widetilde{p}_{E_{1:L, i}}, \widetilde{p}_{Z_{1:i-1}^{1:N}} \widetilde{p}_{E_{1:L, i}} \right)\\
            & =  \mathbb{V} \left( \widetilde{p}_{Z_{1:i-1}^{1:N}E_{1:L, i}} , \widetilde{p}_{Z_{1:i-2}^{1:N}} \widetilde{p}_{ Z_{i-1}^{1:N}E_{1:L, i} } \right)\\\nonumber & \phantom{--}
         +\mathbb{V} \left(  \widetilde{p}_{Z_{i-1}^{1:N} E_{1:L, i}}  , \widetilde{p}_{Z_{i-1}^{1:N}} \widetilde{p}_{E_{1:L, i} }\right)\\\nonumber & \phantom{--}+ \mathbb{V} \left( \widetilde{p}_{Z_{1:i-2}^{1:N}} \widetilde{p}_{Z_{i-1}^{1:N}} , \widetilde{p}_{Z_{1:i-1}^{1:N}}  \right)\\
    & \stackrel{(b)}{\leq}   \delta^{*(1)}_i(N) + 2\mathbb{V} \left( \widetilde{p}_{Z_{1:i-1}^{1:N} E_{1:L, i-1:i}} , \widetilde{p}_{Z_{1:i-2}^{1:N}} \widetilde{p}_{ Z_{i-1}^{1:N}  E_{1:L, i-1:i} } \right)
          \\
            & \stackrel{(c)}{=}  \delta^{*(1)}_i(N)
         +2 \mathbb{V} \left( \widetilde{p}_{Z_{1:i-2}^{1:N} E_{1:L, i-1}} \widetilde{p}_{Z_{i-1}^{1:N} E_{1:L, i} | E_{1:L, i-1}}  ,\right.\\\nonumber & \phantom{--}\left. \widetilde{p}_{Z_{1:i-2}^{1:N}} \widetilde{p}_{ Z_{i-1}^{1:N} E_{1:L, i-1:i}} \right)\\
            & =   \delta^{*(1)}_i(N)+ 2 \mathbb{V} \left( \widetilde{p}_{Z_{1:i-2}^{1:N}E_{1:L, i-1}}   , \widetilde{p}_{Z_{1:i-2}^{1:N}} \widetilde{p}_{   E_{1:L, i-1}} \right)\\
                          & \stackrel{(d)}{\leq}  \delta^{*(1)}_i(N)  + 2 \delta_{i-1}^{*(2)}(N),
  \end{align*}
 where $(a)$ holds by the triangle inequality, $(b)$  holds  by Lemma~\ref{c2lem6},   $(c)$ holds by the  Markov chain $(\widetilde{E}_{1:L, i}, \widetilde{Z}_{i-1}^{1:N}) - \widetilde{E}_{1:L, i-1} - \widetilde{Z}_{1:i-2}^{1:N}$, $(d)$ holds by the induction hypothesis. 
\end{proof}
The following lemmas show that the channel outputs of all the blocks are asymptotically independent, and  that the  target  output  distribution  is  well approximated jointly over all blocks.
    \begin{lem}
    We have
$$
\mathbb{V} \left( \widetilde{p}_{Z_{1:k}^{1:N}},\prod_{i =1}^k  \widetilde{p}_{Z_{i}^{1:N}} \right)\leq (k-1)  \delta_k^{*(2)} (N),
$$ where $\delta_k^{*(2)} (N)$ is defined in Lemma \ref{c2lemnew}. \label{c2lem8}
  \end{lem}
             \begin{lem} \label{c2lem9}  For block $i\in \llbracket 1, k \rrbracket$,
         we have
$$\mathbb{V} \left( \widetilde{p}_{Z_{1:k}^{1:N}},   q_{Z^{1:kN}}  \right) \leq    (k-1)  \delta_k^{*(2)} (N)+ k  \delta^{*}_{k}(N),
$$    where $\delta_k^{*(2)} (N) $ is defined in Lemma \ref{c2lemnew} and $\delta_k^{*}(N)$ is defined in Lemma \ref{lemblock2}.  \end{lem}
   The proofs of Lemmas \ref{c2lem8} and \ref{c2lem9} are similar to the proofs of Lemmas~\ref{lem8} and~\ref{lemindep}, respectively, and are thus omitted. Finally, the next lemma shows that the encoding scheme of Section~\ref{sec:mult} achieves the desired rate-tuple.
\begin{lem}
   Let $\epsilon_0>0$. For $k$ large enough and any~$l \in \mathcal{L}$, we have
   $
               \displaystyle  \lim_{{{N \rightarrow +\infty}}}R_{l} = I(X_{l};Z|X_{1:l-1})+\epsilon_0+2\xi$.\label{c2lem10}
   \end{lem} 
   \begin{proof}
   Let $k$ be such that for any $l \in \mathcal{L}$ we have $
      \frac{
      H(X_l)}{k}  <\epsilon_0.$ 
Then, by the definition of $\epsilon_2$, for any $l \in \mathcal{L}$, we have
         \begin{align*}
          R_l &= \frac{ \sum_{i=1}^k | E_{l, i}|}{kN} \displaybreak[0]\\ 
          &=  \frac{N(H(X_l)+\epsilon_2)  + (k-1) N(I(X_l;ZX_{1:l-1})+\epsilon_2)}{kN}\\
  & \leq    \frac{H(X_l)}{k}+ I(X_l;ZX_{1:l-1})+\epsilon_2 \\
 &  \leq   \epsilon_0+I(X_l;ZX_{1:l-1})+\epsilon_2\\
   & \xrightarrow{N \to +\infty}I
   (X_l;ZX_{1:l-1})+\epsilon_0+2\xi. 
      \end{align*} 
   \end{proof}
      \section{Concluding Remarks}\label{sec:conclusion} 
   We showed that codes for MAC resolvability can be obtained solely from source resolvability codes, used as black boxes, and two-universal hash functions.    The crux of our approach is randomness recycling implemented with distributed hashing across a block-Markov coding scheme.   
   Since explicit constructions for source resolvability codes and two-universal hash functions are known, our approach provides explicit codes to achieve the entire multiple access channel resolvability region for arbitrary channels with binary input~alphabets. 
    \appendices
    \section{An explicit coding scheme for source resolvability} \label{Appsc}
 Let $n \in \mathbb{N}$ and $N \triangleq 2^n$. Let $G_n \triangleq  \left[ \begin{smallmatrix}
       1 & 0            \\[0.3em]
       1 & 1 
     \end{smallmatrix} \right]^{\otimes n} $ be the source polarization matrix defined in \cite{arikan2010source}.  For any set $\mathcal{A} \subseteq \llbracket 1,N \rrbracket$ and any sequence $X^{1:N}$, let $X^{1:N}[\mathcal{A}]$ be the components of $X^{1:N}$ whose indices are in~$\mathcal{A}$. 
         Next, consider a binary memoryless source $(\mathcal{X},q_X)$, where $|\mathcal{X}|=2$. Let $X^{1:N}$ be distributed according to $q_{X^{1:N}} \triangleq \prod_{i=1}^N q_X$, and define $A^{1:N} \triangleq G_n X^{1:N}$. Define also for $\beta < 1/2$, $\delta_N \triangleq 2^{-N^{\beta}}$, the~sets
\begin{align*}
    \mathcal{V}_X &\triangleq \left\{ i \in \llbracket 1, N \rrbracket: H( A^i | A^{1:i-1}) >    1 - \delta_N \right\} , \\
        \mathcal{H}_X &\triangleq \left\{ i \in \llbracket 1, N \rrbracket: H( A^i | A^{1:i-1}) >     \delta_N \right\}.
\end{align*}

\begin{algorithm}[h]
  \caption{Encoding algorithm for source resolvability}
  \label{alg:sr}
  \begin{algorithmic}   [1] 
     \REQUIRE A vector $R$ of $ |\mathcal{V}_X|$  uniformly distributed bits
     \STATE Define $\widetilde A^{1:N} [\mathcal{V}_X] \triangleq R  $
     \STATE  Define $\widetilde A^j$ according to $q_{A^j|A^{1:j-1}}$ for  $j\in \mathcal{V}_X^c \backslash \mathcal{H}_X^c$ and as $\widetilde A^j \triangleq \displaystyle\argmax_{a \in \{0,1\}} q_{A^j|A^{1:j-1}}\!(a|a^{1:j-1}\!)$ for $j\in \mathcal{H}_X^c$
      \STATE Define $\widetilde X^{1:N} \triangleq \widetilde A^{1:N} G_n$
       \end{algorithmic}
\end{algorithm}
  In Algorithm \ref{alg:sr}, the distribution of $\widetilde X^{1:N}$ is such that  $ \lim_{N \to \infty} \mathbb{V}(\widetilde{p}_{X^{1:N}},q_{X^{1:N}}) =0$ by \cite{chou2016polar,chou2015using}. Moreover, the rate of $R$ is $\frac{|\mathcal{V}_X|}{N} \xrightarrow{N \to +\infty} H(X)$ by \cite[Lemma 1]{chou2013polar}, and the rate of randomness used in Line 2 is $0$ by~\cite[Lemma~20]{chou2018empirical}. Hence, Algorithm \ref{alg:sr} achieves the source resolvability of $(\mathcal{X},q_X)$.
    
   \section{Supporting Lemmas}  \label{appA} 
    
 A function $f_X$ defined over a finite alphabet $\mathcal{X}$ is  subnormalized non-negative if $f_X(x)\geq 0 , \forall x \in \mathcal{X}$ and $\sum_{x\in\mathcal{X}}f_X(x) \leq 1$. Additionally, for a subnormalized non-negative function $f_{XY}$ defined over a finite alphabet $\mathcal{X}\times \mathcal{Y}$, its marginals are defined as $f_X(x) \triangleq \sum_{y\in\mathcal{Y}} f_{XY}(x,y), \forall x \in \mathcal{X}$ and $f_Y(y)\triangleq \sum_{x\in\mathcal{X}} f_{XY}(x,y), \forall y \in \mathcal{Y}$, similar to probability distributions.
  
    \begin{lem} [{\cite{chou2018a},\cite[Lemma 2]{chou2021distributed}}] \label{lems1}
       Define $\mathcal{A}\triangleq \llbracket 1,A \rrbracket $.  Let $(\mathcal{T}_a)_{a \in \mathcal{A}}$ be $A$ finite alphabets and define for $\mathcal{S}\subseteq \mathcal{A}$, $\mathcal{T}_{\mathcal{S}}\triangleq \bigtimes_{a \in \mathcal{S}} \mathcal{T}_a$. Consider the random variables  $T^{1:N}_{\mathcal{A}}\triangleq ({T}^{1:N}_a)_{a \in \mathcal{A}}$ and $Z^{1:N}$ defined over $\mathcal{T}_{\mathcal{A}}^N  \times\mathcal{Z}^N$ with probability distribution $q_{T^{1:N}_{\mathcal{A}} Z^{1:N}}\triangleq \prod_{i=1}^N q_{T_{\mathcal{A}} Z}$. For any $\epsilon>0$, there exists a subnormalized non-negative function $w_{T^{1:N}_{\mathcal{A}} Z^{1:N}}$ defined over $\mathcal{T}^N_{\mathcal{A}} \times\mathcal{Z}^N$ such that $\mathbb{V}(q_{T^{1:N}_{\mathcal{A}} Z^{1:N}},w_{T^{1:N}_{\mathcal{A}} Z^{1:N}})\leq\epsilon$  and
       \begin{align*}
          H_{\infty}(w_{T^{1:N}_{\mathcal{S}} Z^{1:N}}|q_{Z^{1:N}})\geq N H({T_{\mathcal{S}}}|Z)-N \delta_{\mathcal{S}}(N), \forall \mathcal{S}\subseteq \mathcal{A},
       \end{align*}
       where $\delta_{\mathcal{S}}(N)\triangleq \log (\lvert\mathcal{T}_{\mathcal{S}}\rvert+3)\sqrt{\frac{2}{N}(A-\log \epsilon)}$, and we have defined the min-entropy as in \cite {renner2008security,watanabe2013non}, i.e.,
       \begin{align*}
          & H_{\infty}(w_{T^{1:N}_{\mathcal{S}} Z^{1:N}}|q_{Z^{1:N}})\\ & \triangleq  -\log \!\!\!\displaystyle \max_{\substack{{t^{1:N}_{\mathcal{S}} \in \mathcal{T}^N_{\mathcal{S}}}\\ z^{1:N} \in \textup{supp}(q_{Z^{1:N}}) }}\!\!\frac{w_{T^{1:N}_{\mathcal{S}} Z^{1:N}}(t^{1:N}_{\mathcal{S}}, z^{1:N})}{q_{Z^{1:N}}(z^{1:N})}.\label{lem2}
       \end{align*}
       \end{lem}
       
     \begin{lem}[{\cite{chou2018a},\cite[Lemma 1]{chou2021distributed}}]
    Consider a sub-normalized non-negative function  $ p_{ X_{\mathcal L}Z}$ defined over $\bigtimes_{l \in \mathcal{L}}\mathcal{X}_{l}\times \mathcal{Z}$, where $X_{\mathcal{L}} \triangleq (X_l)_{l \in \mathcal{L}}$ and, $\mathcal{Z}$, $\mathcal{X}_{l}$, $l \in\mathcal{L}$, are finite alphabets. For $l \in \mathcal{L}$, let  $F_l:\{0,1\}^{n_l} \longrightarrow \{0,1\}^{r_l}$, be uniformly chosen in a family $\mathcal{F}_l$ of two-universal hash functions. Define $s_{\mathcal L} \triangleq \prod_{l \in \mathcal{L}} |\mathcal{F}_l|$, and for any $\mathcal{S} \subseteq \mathcal{L}$, define $r_{\mathcal{S}}\triangleq \sum_{i \in \mathcal{S}}r_i$. Define also ${F}_{\mathcal{L}}\triangleq (F_l)_{l \in \mathcal{L}}$ and 
  $
        F_{\mathcal{L}}( X_{\mathcal{L}})\triangleq \left( F_l(X_l) \right)_{l \in \mathcal{L}}.
   $\label{lem11}
   Then, for any $q_Z$ defined over $\mathcal{Z}$ such that $\textup{supp}(q_Z) \subseteq \textup{supp}(p_Z)$, we have
   \begin{align*}
     & \mathbb{V}({{p}_{F_{\mathcal{L}}( X_{\mathcal{L}}),   F_{\mathcal{L}},Z}}, p_{U_{\mathcal K}} p_{U_{\mathcal F}} p_Z) \\&\leq   {{{\sqrt{{ \displaystyle\sum_{\substack{{\mathcal S\subseteq\mathcal L}, {\mathcal S \neq \emptyset}}}}2^{r_{\mathcal S}-H_{\infty}(p_{X_{\mathcal S}Z}|q_Z)}}}}},
       \end{align*}
      where  $p_{U_{\mathcal K}}$ and  $\textcolor{black}{p_{U_{\mathcal F}}}$ are the uniform distributions over $\llbracket 1,2^{r_{{\mathcal{L}}}} \rrbracket$ and $\llbracket 1,{s_{{\mathcal{L}}}} \rrbracket$, respectively.
          \end{lem}
    \section{Proof of Lemma \ref{lem1}} \label{AppendixA}
    The proof is similar to \cite{grant2001rate}. We have
\begin{align*}
    I(XY;Z) 
 &\stackrel{(a)} =I(XUV;Z)\\
&\stackrel{(b)} =I(U;Z)+I(X;Z\lvert U)+I(V;Z\lvert UX),
\end{align*} 
where $(a)$ holds because $I(XUV;Z)\geq I(XY;Z)$ since $Y=f(U,V)$, and $I(XUV;Z)\leq I(XY;Z)$ since  $(X,U,V)-(X,Y)-Z$ forms a Markov chain, $(b)$ holds by the chain~rule.

We  know  by {\cite [Lemma 6]{grant2001rate}}  that $I(X;ZU)$ is a continuous function of $\epsilon$,  hence so is
\begin{align*}
    R_1=I(X;Z|U)=I(X;ZU),
\end{align*}
    where the last equality holds by the independence between $X$ and $U$. Then, $I(X;Z)$ and $I(X;Z|Y)$ are in the image of $R_1$ by $\eqref{lem1eqn1}$, and hence, using $I(X;Z)\leq I(X;YZ)=I(X; Z|Y)$, $ [ I(X;Z), I(X;Z|Y)]$ is also in the image of $R_1$ by continuity.

          \bibliographystyle{IEEEtran}
               \bibliography {main_v3}

\end{document}